\documentclass[11pt]{article}

\usepackage{lineno,hyperref}
\usepackage
[
letterpaper,
left=1.91cm,
right=1.91cm,
top=1.91cm,
bottom=2cm,
]
{geometry}   
\modulolinenumbers[5]
\usepackage{amsfonts,amsmath,amstext,amssymb,amsthm}
\usepackage{graphicx}         
\usepackage{float}
\usepackage{subcaption}
\usepackage{array}
\usepackage{xcolor}
\usepackage{hyperref}
\usepackage{algorithm, algpseudocode}
\usepackage{bigints}
\usepackage[algo2e,vlined,linesnumbered]{algorithm2e}

\newcommand{\revise}[1]{{\color{black}  #1}}
\newcommand\norm[1]{\left\lVert#1\right\rVert}
\newcommand{\bmat}[1]{\begin{bmatrix}#1\end{bmatrix}}
\newcommand{\R}{{\mathbb{R}}}
\newcolumntype{M}[1]{>{\centering\arraybackslash}m{#1}}
\newtheorem{theorem}{Theorem}
\newtheorem{example}{Example}
\newtheorem{remark}{Remark}
\newtheorem{lemma}{Lemma}
\newtheorem{proposition}{Proposition}
\newtheorem{assumption}{Assumption}

\newtheorem{definition}{Definition}

\title{\LARGE \bf Reachability Analysis Using Dissipation Inequalities for Uncertain Nonlinear Systems}

\author{
	He Yin \thanks{He Yin is a Graduate Student in the Department of 
		Mechanical Engineering at the University of California, Berkeley 
		{\tt\small he\_yin@berkeley.edu}} , 
	Andrew Packard \thanks{Andrew Packard is a Professor in the Department of
		Mechanical Engineering at the University of California, Berkeley
		{\tt\small apackard@berkeley.edu}} ,
	Murat Arcak \thanks{Murat Arcak is a Professor in the Department of Electrical Engineering
		and Computer Sciences at the University of California, Berkeley
		{\tt\small arcak@berkeley.edu}} ,
	Peter Seiler \thanks{Peter Seiler is an Associate Professor in the Department of Electrical Engineering and Computer Science, University of Michigan, Ann Arbor
		{\tt\small pseiler@umich.edu}}%
}
\date{\vspace{-5ex}}

\begin{document}
\maketitle

\begin{abstract}
Abstract --- We propose a method to outer bound forward reachable sets on finite horizons for uncertain nonlinear systems with polynomial dynamics. This method makes use of time-dependent polynomial storage functions that satisfy appropriate dissipation inequalities that account for time-varying uncertain parameters, $\mathcal{L}_2$ disturbances, and perturbations $\Delta$ characterized by integral quadratic constraints (IQCs) with both hard and soft factorizations. In fact, to our knowledge, this is the first result introducing IQCs to reachability analysis, thus allowing for various types of uncertainty, including unmodeled dynamics. The generalized S-procedure and Sum-of-Squares techniques are used to derive algorithms with the goal of finding the tightest outer bound with a desired shape. Both pedagogical and practically motivated examples are presented, including a 7-state F-18 aircraft model.
\end{abstract}



\section{Introduction}
The forward reachable set (FRS) is the set of all the successors to a set of initial conditions subject to the given dynamics under all possible model uncertainties and disturbances in a finite horizon. The computation of the FRS plays an important role in safety-critical systems, as it can verify whether a system is able to reach a target and avoid an obstacle \cite{Prajna:04}, \cite{Mitchell:00}. Indeed, if an outer bound avoids obstacles and is encompassed by the target set at the final time, then one can ascertain the same properties for all trajectories. In this paper, we present a method for finding the smallest achievable 
outer bounds to the FRSs on finite horizons, since in many practical settings, systems only undergo finite-time trajectories, such as robotic systems and space launch and re-entry vehicles.

The algorithm proposed in this paper uses a storage function that satisfies a dissipation inequality to characterize the outer bound. The dissipation inequality framework, combined with the Sum-of-Squares (SOS) technique~\cite{Antonis:02} and the generalized S-procedure \cite{Parrilo:00}, allows us to simultaneously accommodate multiple sources of uncertainty, including time varying uncertain parameters, $\mathcal{L}_2$ disturbances, and perturbations $\Delta$ whose input output properties are characterized by integral quadratic constraints (IQCs) \cite{Megretski:97}. IQCs can model a rich class of uncertainties and nonlinearities, including hard nonlinearities (e.g. saturation), time delays, and unmodeled dynamics, as summarized in \cite{Megretski:97} and \cite{Veenman:16}. Therefore, although our nominal systems are assumed to be polynomials, including  IQCs allows us to extend our analysis framework to the class of systems beyond polynomial systems. IQCs are also used in \cite{Chris:19} for robustness analysis of linear time varying systems, and in \cite{IANNELLI:2019} for region of attraction analysis of nonlinear systems.

The proposed analysis framework considers both hard and soft IQC factorizations. Dissipation inequalities usually require IQCs to hold over all finite horizons (hard IQCs) \cite{Summers:13}, \cite{Balas:11}. However, IQCs are often available in the infinite-time horizon (soft IQCs), while the hard IQCs are not. To mitigate this issue, we incorporate soft IQCs in dissipation inequalities by making use of a lower bound derived from \cite{Fetzer:2018}, which is valid for soft IQCs over all finite horizons.
 We formulate the reachable set computation as generalized SOS optimization problems that are quasi-convex, which can be solved effectively by bisection. In addition, our optimization problems do not require a feasible initialization of the storage function. 

There are various existing approaches to reachability analysis, including interval analysis \cite{Jaulin:01}, Hamilton-Jacobi methods \cite{Mitchell:00}, ellipsoid methods \cite{Varaiya:07} and polytope methods \cite{Borrelli:11}. Dissipation inequalities and SOS programming were introduced to reachability analysis in \cite{Packard:05}, \cite{Weehong:06}, \cite{Tan:08}, and further extended to uncertain nonlinear systems in \cite{Topcu:09}. However, these previous results do not consider integral quadratic constraints (IQCs) and they bound the reachable sets on the infinite-time horizon, which might yield overly conservative estimates for finite-time trajectories. 

Finite-time reachability analysis using SOS programming is considered in \cite{Majumdar:17}, \cite{Tobenkin:11} using Lyapunov-based method. Another finite-time reachability analysis paper, \cite{Bai:18}, provides an approximate analytical polynomial solution to the Hamilton-Jacobi-Isaacs partial differential equations (HJE) for reachable set computation. 
The works \cite{xue2019inner, xue2019robust} extend \cite{Bai:18} to the state-constrained polynomial systems with time-varying uncertainties. While the work \cite{xue2019inner} focuses on inner-approximating the finite-time horizon backward reachable set (BRS), the work \cite{xue2019robust} focuses on computing the infinite-time horizon robust invariant sets. Occupation measures based method is proposed in \cite{henrion2013convex} to compute outer-approximations to the BRS for polynomial systems with control inputs. Note that \cite{Majumdar:17, Tobenkin:11, xue2019inner, xue2019robust,henrion2013convex} only allow for time varying parametric uncertainty. In contrast, our finite horizon results account for various sources of uncertainty. To our knowledge, this is the first paper that introduces IQCs to reachability analysis. 

To summarize, the main contributions of the paper are: (i) to analyze finite-time horizon reachability with robustness guarantees,  (ii) to extend the framework to a large class of uncertain systems by incorporating IQCs.
	
The paper is organized as follows. Section \ref{sec:reach_l2} presents the problem setup, the basic theorem, and computation method for outer bounding the reachable sets for the nominal system: nonlinear system with $\mathcal{L}_2$ disturbances and time varying uncertain parameters. Sections~\ref{sec:reach_perturb} considers the robust reachability analysis for the uncertain system: interconnection of nominal system and perturbations $\Delta$ described by hard IQCs. Sections~\ref{sec:soft_IQC} extends the robust analysis framework to $\Delta$ that satisfies soft IQCs. Section \ref{sec:examples} applies the method to several aircraft examples, one of which is compared with the result obtained using the method from \cite{Majumdar:17}. Section \ref{sec:conclu} summarizes the results.

\subsection{Notation}
$\mathbb{R}^{m\times n}$ and $\mathbb{S}^n$ denote the set of $m$-by-$n$ real matrices and $n$-by-$n$ real, symmetric matrices. 
$\mathbb{RL}_{\infty}$ is the set of rational functions with real coefficients that have no poles on the imaginary axis. $\mathbb{RH}_{\infty} \subset \mathbb{RL}_{\infty}$ contains functions that are analytic in the closed right-half of the complex plane.  $\mathcal{L}_2^{n_r}$ is the space of measurable functions $r: [0, \infty) \rightarrow \mathbb{R}^{n_r}$, with $\norm{r}^2_2 := \int_0^{\infty} r(t)^\top r(t) dt < \infty$. Associated with $\mathcal{L}_2^{n_r}$ is the extended space $\mathcal{L}_{2e}^{n_r}$, consisting of functions whose truncation $r_T(t) := r(t)$ for $t \le T$; $r_T(t) := 0$ for $t > T$, is in $\mathcal{L}_2^{n_r}$ for all $T > 0.$ 
Define the finite-horizon $\mathcal{L}_2$ norm as $\norm{r}_{2,[t_0,T]}:= \left(\int_{t_0}^T r(t)^\top r(t)dt \right)^{1/2}$. If $r$ is measurable, and $\norm{r}_{2,[t_0,T]} < \infty$ then $r \in \mathcal{L}^{n_r}_2[t_0,T]$. The finite horizon induced $\mathcal{L}_2$ to $\mathcal{L}_2$ norm is denoted as $\norm{\cdot}_{2 \rightarrow 2, [t_0,T]}$.
For $\xi \in \mathbb{R}^n$, $\mathbb{R}[\xi]$ represents the set of polynomials in $\xi$ with real coefficients, and $\R^m[\xi]$ and $\R^{m \times p}[\xi]$ to denote all vector and matrix valued polynomial functions. The subset $\Sigma[\xi] := \{\pi = \pi_1^2 + \pi_2^2 + ... + \pi_M^2 : \pi_1, ..., \pi_M \in \mathbb{R}[\xi]\}$ of $\mathbb{R}[\xi]$ is the set of SOS polynomials in $\xi$. 
For $\eta \in \mathbb{R}$, and continuous $r: \mathbb{R}^n \rightarrow \mathbb{R}$, $\Omega_{\eta}^r := \{x \in \mathbb{R}^n : r(x) \le \eta\}.$ For $\eta \in \mathbb{R}$, and continuous $g : \mathbb{R} \times \mathbb{R}^{n} \rightarrow \mathbb{R}$, define $\Omega_{t,\eta}^{g} := \{x \in \mathbb{R}^n : g(t,x) \le \eta\}$, a $t$-dependent set. $KYP$ denotes a mapping to the block 2-by-2 matrix:
\begin{align}
KYP(Y,A,B,C,D,M) := \bmat{A^\top Y + YA & YB \\ B^\top Y & 0} + \bmat{C^\top \\ D^\top} M \bmat{C & D}. \label{eq:kyp_def}
\end{align}

\section{Nominal Reachability Analysis} \label{sec:reach_l2}
Consider the nominal nonlinear system $N$ defined on $[t_0, T]$:
\begin{equation}
\dot{x}(t) = f(t, x(t), w(t), \delta(t)), \label{eq:nominal_sys}
\end{equation}
where $x(t) \in \mathbb{R}^n$ is the state, $w(t) \in \mathbb{R}^{n_w}$ is the external disturbance, $\delta(t) \in \R^{n_\delta}$ is the time varying uncertain parameter, and $f : \mathbb{R} \times \mathbb{R}^n \times \mathbb{R}^{n_w} \times \R^{n_\delta}\rightarrow \mathbb{R}^n$ is the vector field. 
\begin{assumption}\label{ass:ass1}
(i) the disturbance $w$ satisfies $w \in \mathcal{L}_{2}^{n_w}$ with $\norm{w}_{2,[t_0, T]} < R$ for some $R>0$, (ii) there exists a non-decreasing polynomial function $h: \mathbb{R} \rightarrow \mathbb{R}_{\ge 0}$ with $h(t_0) = 0$, $h(T) = 1$ such that
\begin{align}
\int_{t_0}^t w(\tau)^\top w(\tau) d\tau < R^2 h(t), \ \forall t \in [t_0, T], \label{eq:w_rate}
\end{align}
(iii) the function $\delta: [t_0, T] \rightarrow \R^{n_\delta}$ is measurable, and for each $t \in [t_0, T]$, $\delta(t) \in \mathcal{D} := \{\delta \in \R^{n_\delta} : p_\delta (\delta) \ge 0 \}$.
\end{assumption}

The function $h$ is used to describe how fast the energy of $w$ can be released on the interval $[t_0,T]$. Moreover, the polyomial $p_\delta \in \R[\delta]$ describes the prior knowledge that bounds the uncertainty $\delta$. Next, the definition of the forward reachable set (FRS) is given as follows: 
\begin{definition}
	Under Assumption~\ref{ass:ass1}, the FRS of \revise{the system $N$} \eqref{eq:nominal_sys} from $\mathcal{X}_0$ at time $T$ is defined as
	\begin{align}
	FRS(T;\revise{N},t_0,\mathcal{X}_0,R,h,\mathcal{D}) &:= \{x(T) \in \R^n \ :  \ \exists x(t_0) \in \mathcal{X}_0, w \ \text{satisfying} \ \eqref{eq:w_rate} \ \text{and} \ \delta(t) \in \mathcal{D}, \nonumber \\ 
	 &~~~~~~~~~~~~~~~~~~~~~~~~~~~~~~~~~~~~\revise{\text{such that} \ x(T) \ \text{is a solution to} \ \eqref{eq:nominal_sys} \ \text{at time} \ T} \}. \nonumber
	\end{align}
\end{definition}

Our goal is to outer bound this FRS, and the following theorem provides a way of achieving it based on dissipation-inequalities.

\begin{theorem} \label{thm3}
	Let Assumption~\ref{ass:ass1} hold. Given vector field $f:\R \times \R^n \times \R^{n_w} \times \R^{n_\delta} \rightarrow \R^n$, time interval $[t_0, T]$, local region $\mathcal{X}_l \subset \mathbb{R}^n$,  set of initial conditions $\mathcal{X}_0 \subset \mathbb{R}^n$, disturbance bound $R$, function $h$, and set of uncertain parameters $\mathcal{D}$, suppose there exists a $\mathcal{C}^1$ function $V: \mathbb{R} \times \mathbb{R}^n \rightarrow \mathbb{R}$ that satisfies
	\begin{subequations}\label{eq:constraintD}
	\begin{align}
	&\frac{\partial V(t,x)}{\partial t} + \frac{\partial V(t,x)}{\partial x} f(t,x,w,\delta) \le w^\top w, \ \ \forall (t, x, w, \delta) \in [t_0, T]\times \mathcal{X}_l  \times \mathbb{R}^{n_w} \times \mathcal{D},  \label{eq:constraintD1}\\
	&\mathcal{X}_0 \subseteq \Omega_{t_0,0}^V , \label{eq:constraintD2}\\
	&\Omega_{t,R^2h(t)}^V \subseteq \mathcal{X}_l , \ \ \forall t \in [t_0, T]. \label{eq:constraintD3}
	\end{align}
	\end{subequations}
	Then $x(T) \in \Omega_{T,R^2}^V$ for all $x(t_0) \in \mathcal{X}_0$, \revise{where $x(T)$ is a solution to the system $N$ \eqref{eq:nominal_sys} at time $T$ from $x(t_0)$}. Therefore $\Omega_{T,R^2}^V$ is an outer bound to the $FRS(T;\revise{N},t_0,\mathcal{X}_0,R,h,\mathcal{D})$.
\end{theorem}
\begin{proof}
	Combining constraints (\ref{eq:constraintD1}) and (\ref{eq:constraintD3}), we have the following dissiaption inequality:
	\begin{align}
	&\frac{\partial V(t,x)}{\partial  t} + \frac{\partial V(t,x)}{\partial x}f(t, x, w,\delta) \le w^\top w, \ \ \forall (t, x, w, \delta), \ \text{s.t.} \ t \in [t_0, T], \ x \in \Omega^V_{R^2h(t)}, \ w \in \mathbb{R}^{n_w}, \ \delta \in \mathcal{D}. \nonumber
	\end{align}
	Since this dissipation inequality only holds on the set $\Omega_{t,R^2h(t)}^V$, we need to first prove that all the states starting from $\mathcal{X}_0$ won't leave $\Omega_{t,R^2h(t)}^V$, for all $t \in [t_0, T]$. Assume there exist a time instance $T_1 \in [t_0,T]$, $x_0 \in \mathcal{X}_0$, and signals $w$ satisfying \eqref{eq:w_rate}, $\delta(t) \in \mathcal{D}$, such that a trajectory of the system $N$ starting from $x(t_0) = x_0$ satisfies $V(T_1,x(T_1)) > R^2 h(T_1)$. Define $T_2 = \inf_{V(t,x(t))>R^2 h(t)} t$. Therefore, the dissipation inequality holds on $[t_0, T_2]$, and we can integrate it over $[t_0,T_2]$:
	\begin{align}
	V(T_2,x(T_2)) - V(t_0,x(t_0))&\leq  \int_{t_0}^{T_2} w(t)^\top w(t) dt. \nonumber  \\
	\intertext{By assumption $x_0 \in \mathcal{X}_0$, it follows from \eqref{eq:constraintD2} that $V(t_0,x(t_0)) \leq 0$. Combing it with $w$ satisfying \eqref{eq:w_rate} to show}
	R^2 h(T_2) = V(T_2,x(T_2)) & <  R^2h(T_2). \nonumber 
	\end{align}
	This is contradictory. Therefore there doesn't exist a $T_1 \in [t_0,T]$, such that $x(T_1) \notin \Omega_{T_1,R^2 h(T_1)}^V$. As a result, for all $x(t_0) \in \mathcal{X}_0$, we have $x(t) \in \Omega_{t,R^2h(t)}^V$, for all $t \in [t_0, T]$, and thus $x(T) \in \Omega_{T,R^2}^V$. 
\end{proof}

If the function $h$ is not given, then there is no \textit{a priori} knowledge on how $\int_{t_0}^t w(\tau)^\top w(\tau) d\tau$ depends on $t$. In this case the constraint (\ref{eq:constraintD3}) is modified to be
\begin{align}
\Omega_{t,R^2}^V \subseteq \mathcal{X}_l , \ \ \forall t \in [t_0, T]. \label{eq:relax2}
\end{align}
This case is more restrictive for the storage function and yields larger outer bounds on the FRS.

\revise{We are interested in a tight outer bound to the FRS.  Thus it is natural to search for a storage function $V$ that minimizes the volume of $\Omega_{T,R^2}^V$.} However, an explicit expression is not available for the volume of $\Omega_{T,R^2}^V$ for a generic storage function. Instead, we introduce a user-specified shape function $q$ and its corresponding variable sized region $\Omega_{\alpha}^q = \{x\in \mathbb{R}^n: q(x) \le \alpha\}$. The shape function $q$ can be associated with the user's initial guess of the shape of the actual reachable set or can signify the desired shape of the outer bound. An example of $q$ is given in Section \ref{ex1}. The volume of $\Omega_{T,R^2}^V$ can be shrunk, by enforcing 
\begin{align}
\Omega_{T,R^2}^V \subseteq \Omega_{\alpha}^q, \label{eq:constraintD4} 
\end{align} 
while minimizing $\alpha$. For more heuristic metrics for the volume of semi-algebraic sets, the reader is referred to \cite{MPeet:2019}.


\revise{To find a storage function $V$ that satisfies the constraints in \eqref{eq:constraintD} and \eqref{eq:constraintD4}, we leverage SOS programming.} To do so, we assume that $\mathcal{X}_0$ and $\mathcal{X}_l $ are semi-algebraic sets: $\mathcal{X}_0 := \{x \in \mathbb{R}^n : r_0(x) \leq 0\}$, and
	\begin{align}	
	\mathcal{X}_l :=& \revise{\{x \in \mathbb{R}^n : p(x) \leq \eta\},} \label{eq:local_def} 
	\end{align}
	where $r_0$, $p \in \R[x]$ are  specified by the user, \revise{and $\eta \in \R$ is a decision variable that determines that volume of $\mathcal{X}_l$}. Additionally, we restrict the system model, shape function, and storage function to polynomials, i.e., $f \in \R^n[(t,x,w,\delta)]$, $q \in \R[x]$, $V \in \R[(t,x)]$.
	 Also define $g(t) := (t - t_0)(T - t)$, whose value is nonnegative when $t \in [t_0, T].$ The polynomial functions are used to formulate the set containment constraints. With these ideas, sufficient SOS conditions for the set containment constraints (\ref{eq:constraintD}) and (\ref{eq:constraintD4}) are obtained. \revise{Also by choosing $\alpha$ as the cost function,} we obtain the following SOS optimization problem, denoted as $\boldsymbol{sosopt_1(f, p, g, q, r_0, R, h, p_\delta)}$,
\begingroup
\allowdisplaybreaks
\begin{subequations} \label{eq:nominal_sos}
\begin{align} 
\min_{\alpha, \revise{\eta}, s, V} &\ \ \alpha \nonumber \\
\text{s.t.} \ \ \ & s_5 - \epsilon_1 \in \Sigma[x], s_6 - \epsilon_2 \in \Sigma[(x,t)], \epsilon_1 > 0, \epsilon_2 > 0, \nonumber \\
& s_i \in \Sigma[(x,w,\delta, t)], \forall i \in \{1,2,3\}, s_4 \in \Sigma[x], s_7 \in \Sigma[(x,t)], V \in \mathbb{R}[(t,x)], \label{eq:sos1}  \\
& -\left(\frac{\partial  V}{\partial t} + \frac{\partial  V}{\partial x}f- w^\top w \right)  + \revise{(p-\eta)}s_1 - s_2 g - s_3 p_\delta \in \Sigma[(x,w,\delta,t)], \label{eq:sos2} \\
& -V\vert_{t=t_0} + s_4 r_0 \in \Sigma[x],  \label{eq:sos3}\\
& - \revise{(p-\eta)} s_6 + V - R^2 h - s_7 g \in \Sigma[(x,t)],  \label{eq:sos4}\\
& - ( q - \alpha)s_5 + V\vert_{t=T} - R^2 \in \Sigma[x],  \label{eq:sos5}
\end{align}
\end{subequations}
\endgroup
\noindent where $s_i, i \in \{1,...,7\},$ are SOS polynomials, called multipliers, whose coefficients are to be determined, $\epsilon_1$ and $\epsilon_2$ are small positive numbers chosen by the user to guarantee that $s_5$ and $s_6$ cannot take the value of zero. \revise{The optimization $sosopt_1$ is nonconvex as it is bilinear in two groups of decision variables $(\alpha, \eta)$ and $(s_1, s_5, s_6)$. Since we can't bisect on both $\alpha$ and $\eta$ at the same time, we propose Algorithm~\ref{alg:sol_algo} that solves the problem in two steps, and bisects on one decision variable at one step.}
\begin{algorithm2e} [H]
	\KwData{$f, p, g,q,r_0,R,h,p_\delta $}
	
	\textbf{Preparation Step:} solve for $\eta^\star = \arg\min \eta$ s.t. \eqref{eq:sos1}--\eqref{eq:sos4} by bisecting on $\eta$.
	
	\textbf{Main Step:} solve for $\alpha^\star = \arg\min \alpha$ s.t. \eqref{eq:sos1}--\eqref{eq:sos5} by using $\eta = \eta^\star$ and bisecting on $\alpha$. 
	
	\KwResult{Minimized $\alpha^\star$, outer bound $\Omega_{T,R^2}^V$.}
	\caption{Computing the outer bound}
	\label{alg:sol_algo}
\end{algorithm2e}

The first step of Algorithm~\ref{alg:sol_algo} is to find the smallest feasible local region $\mathcal{X}_l$ (with respect to $p$) by setting aside the original objective function and constraint \eqref{eq:sos5}, and minimizing $\eta$. The second is to find the least conservative outer bound. The first and second steps bisect on $\eta$ and $\alpha$, respectively. Each iteration of bisection involves holding $\alpha / \eta$ fixed and solving a feasibility problem, which is a standard semidefinite programming problem and is convex. If the fixed value of $\alpha / \eta$ leads to infeasibility of the problem, then try to solve it with a larger $\alpha / \eta$; otherwise, decrease the value of $\alpha / \eta$.

	\begin{proposition}
		The SOS constraints \eqref{eq:sos2}--\eqref{eq:sos5} are sufficient conditions for (\ref{eq:constraintD}) and (\ref{eq:constraintD4}).
	\end{proposition}
	
	\begin{proof}
		\eqref{eq:sos2} $\Rightarrow$ \eqref{eq:constraintD1}: The proof follows from the generalized S-procedure \cite{Parrilo:00}. In \eqref{eq:sos2}, when ($x, t,\delta$) satisfies $p(x) \leq \eta$ (i.e. $x \in \mathcal{X}_l$), $g(t) \ge 0$ (i.e. $t \in [t_0, T]$), $p_\delta \ge 0$ (i.e. $\delta \in \mathcal{D}$), for the polynomial in \eqref{eq:sos2} to be nonnegative, then $-(\frac{\partial V(t,x)}{\partial t} + \frac{\partial V(t,x)}{\partial x}f(t,x,w,\delta) - w^\top w)$ must be nonnegative. Thus (\ref{eq:sos2}) implies (\ref{eq:constraintD1}). 
		
		\eqref{eq:sos3} $\Rightarrow$ \eqref{eq:constraintD2}:  
		In \eqref{eq:sos3}, when a state $x$ satisfies $r_0(x) \leq 0$ (i.e. $x \in \mathcal{X}_0$), for the polynomial in \eqref{eq:sos3} to be nonnegative, then $-V(t_0, x)$ must be nonnegative (i.e. $x \in \Omega_{t_0, 0}^V$). 
		
		\eqref{eq:sos4} $\Rightarrow$ \eqref{eq:constraintD3}: 
		In (\ref{eq:sos4}), when a state and time pair ($x, t$) satisfies $V(t,x) \leq R^2 h(t)$ (i.e. $x \in \Omega_{t,R^2 h(t)}^{V}$) and $g(t) \ge 0$ (i.e. $t \in [t_0, T]$), for the polynomial in (\ref{eq:sos4}) to be nonnegative, then $-s_6(t,x) (p(x)-\eta)$ must be nonnegative (i.e. $x \in \mathcal{X}_l$).
		
		\eqref{eq:sos5} $\Rightarrow$ \eqref{eq:constraintD4}: 
		In \eqref{eq:sos5}, when a state $x$ satisfies $V(T,x) \leq R^2$ (i.e. $x \in \Omega_{T, R^2}^{V}$), for the polynomial in \eqref{eq:sos5} to be nonnegative, then $-(q(x) - \alpha)s_5(x)$ must be nonnegative (i.e. $x \in \Omega_{\alpha}^q$). 
	\end{proof}

\subsection{Application to a 2-state example} \label{ex1}
Consider the following academic example from \cite{Packard:05}:
\begin{align} \label{eq:Tan_dyn}
\begin{split}
\dot{x}_1 =& -x_1 + x_2 - x_1 x_2^2,  \\
\dot{x}_2 =& -x_2 - x_1^2x_2 + w,
\end{split}
\end{align}
where $w$ is the disturbance satisfies (\ref{eq:w_rate}) with $R = 1$ and $h(t) = t^2/T^2$. In this example the uncertain parameter is not considered. We take $[t_0, T] = [0, 1 \ \text{sec}]$, $r_0(x) = x^\top x - 1$. In Figure \ref{fig:2-state}, the green points are simulation points $x(T)$, at $T = 1$ sec, for the system (\ref{eq:Tan_dyn}) using disturbance signals $w$, with initial conditions inside $\mathcal{X}_0$, which is shown with the red dotted curve. In this example,  the shape function $q$ is obtained by computing the minimum volume ellipsoid $\Omega_1^q$ that contains all the simulation points $x(T)$ at $T = 1$ sec, and $q(x) = 4.84x_1^2 - 3.05x_1x_2 + 1.50x_2^2$. A more accurate shape function can be obtained by fitting a higher degree polynomial to the simulation points \cite{Alessandro:05}. Here, the polynomial $p$ that defines the local region is obtained by computing the minimum volume ellipsoid $\Omega^{p}_1$ that contains some sampling points on the simulation trajectories $x(t)$, for $t \in [0, T]$, and $p = 0.989 x_1^2 - 0.051 x_1 x_2 + 0.949 x_2^2 + 0.001 x_1 + 0.001 x_2$. Solving the first step of Algorithm~\ref{alg:sol_algo} we obtain $\eta^\star = 1.044$, and $\mathcal{X}_l$ is determined.
Solving the second step gives $\alpha^\star = 1.37$. The outer bound is shown with the black curve, which tightly encloses all $x(T)$.
\begin{figure}[h]
	\centering
	\includegraphics[width=0.53\columnwidth]{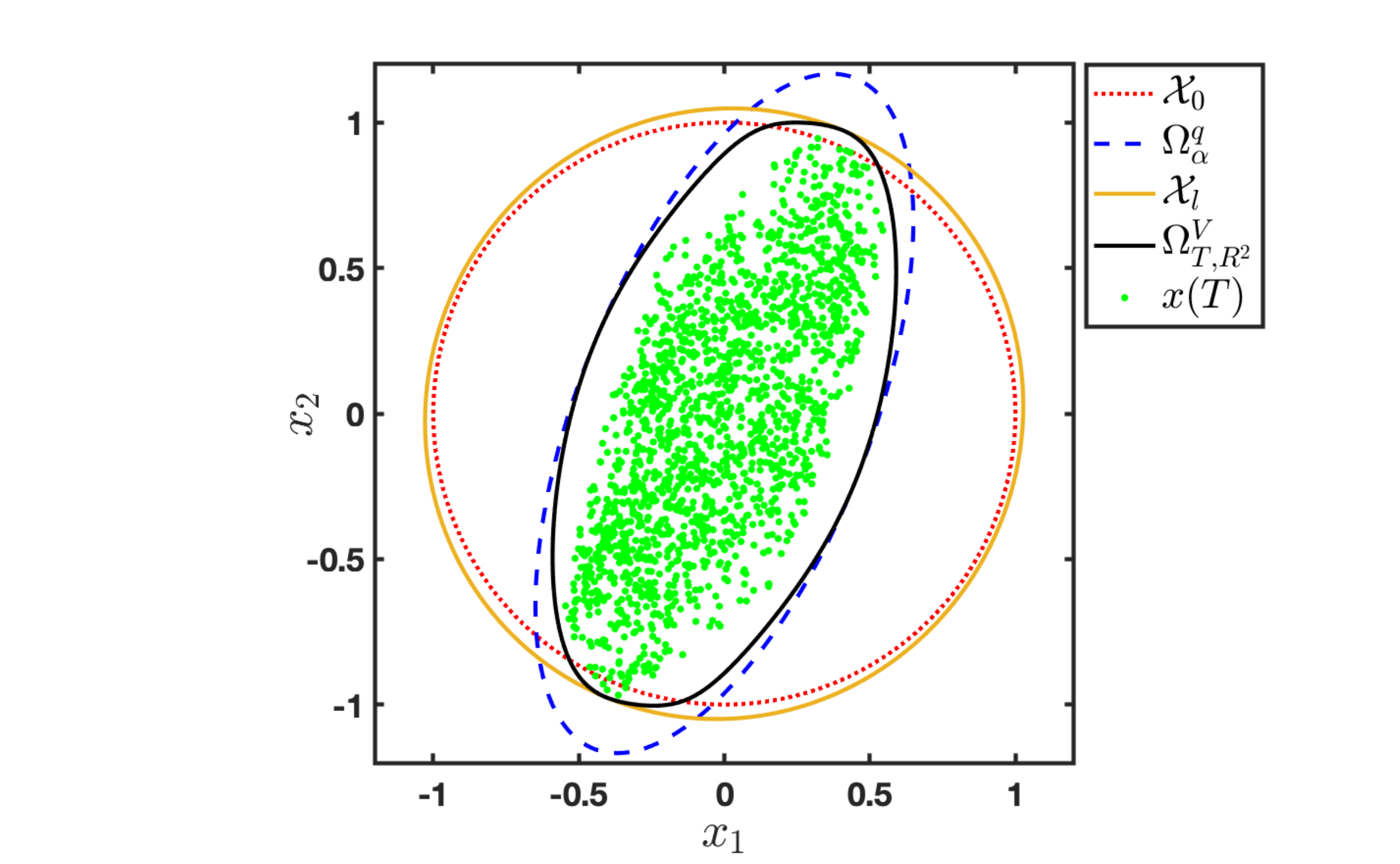}
	\caption{Outer bound of reachable set at $T = 1$ sec for the 2-state example with $\mathcal{L}_2$ disturbance. }
	\label{fig:2-state}    
\end{figure}

\section{Robust Reachability Analysis with Hard IQCs} \label{sec:reach_perturb}
Consider the uncertain nonlinear system shown in Figure \ref{fig:Fu}, which is an interconnection $F_u(G,\Delta)$ of a nonlinear system $G$ and a perturbation $\Delta$. \revise{The dynamics of the nonlinear system $G$ are of the form}:
\begin{align}\label{eq:sysG4IQC}
\begin{split}
\dot{x}_G(t) =& f(t, x_G(t), l(t), w(t), \delta(t)),  \\
v(t) =& r(t, x_G(t), l(t), w(t), \delta(t)), 
\end{split}
\end{align}
where $x_G(t) \in \mathbb{R}^{n_G}$ is the state of $G$, and $\delta(t) \in \R^{n_\delta}$ is the uncertain parameter. \revise{The inputs of $G$ are $w(t) \in \mathbb{R}^{n_w}$ and $l(t) \in \mathbb{R}^{n_l}$, while the output is $v(t) \in \mathbb{R}^{n_v}$. The system $G$ is defined by the mappings} $f : \mathbb{R} \times \mathbb{R}^{n_G} \times \mathbb{R}^{n_l} \times \mathbb{R}^{n_w} \times \R^{n_\delta}$ $\rightarrow \mathbb{R}^{n_G}$ and $r: \mathbb{R} \times \mathbb{R}^{n_G} \times \mathbb{R}^{n_l} \times \mathbb{R}^{n_w}\times \R^{n_\delta} \rightarrow \R^v$. The perturbation is a \revise{bounded and causal} operator $\Delta : \mathcal{L}_{2e}^{n_v} \rightarrow \mathcal{L}_{2e}^{n_l}$. \revise{Assume the interconnection $F_u(G,\Delta)$ formed by $G$ and $\Delta$ through the constraint
\begin{equation}
l(\cdot) = \Delta (v(\cdot))\label{eq:def_Delta}
\end{equation}
is well-posed. The well-posedness of the interconnection $F_u(G,\Delta)$ is defined as follows.}
\begin{definition}
		\revise{$F_u(G,\Delta)$ is well-posed if for all $x_G(t_0) \in \R^{n_G}$ and $w \in \mathcal{L}_{2e}^{n_w}$ there exist unique solutions $x_G \in \mathcal{L}_{2e}^{n_G}$, $v \in \mathcal{L}_{2e}^{n_v}$, and $l \in \mathcal{L}_{2e}^{n_l}$ satisfying \eqref{eq:sysG4IQC} and \eqref{eq:def_Delta} with a causal dependence on $w$.}
\end{definition}

Again, assume all the trajectories of $F_u(G,\Delta)$ start from $x_G(t_0) \in \mathcal{X}_0 \subset \R^{n_G}$. Similarly, the FRS of $F_u(G,\Delta)$ from $\mathcal{X}_0$ at time $T$ is defined as \revise{
\begin{align} \label{def:FRS_Fu}
FRS(T;F_u(G,\Delta),t_0,\mathcal{X}_0,R,h, \mathcal{D}) &:= \{x_G(T) \in \R^{n_G} :  \exists \ x_G(t_0) \in \mathcal{X}_0, w \ \text{satisfying} \ \eqref{eq:w_rate} \ \text{and} \ \delta(t) \in \mathcal{D},  \nonumber \\
&~~~~~~~~~~~~\revise{\text{such that} \ x_G(T) \ \text{is a solution to} \ \eqref{eq:sysG4IQC} - \eqref{eq:def_Delta} \ \text{at time} \ T} \}.  
\end{align}   }
\begin{figure}[h]
	\centering
	\includegraphics[width=0.3\textwidth]{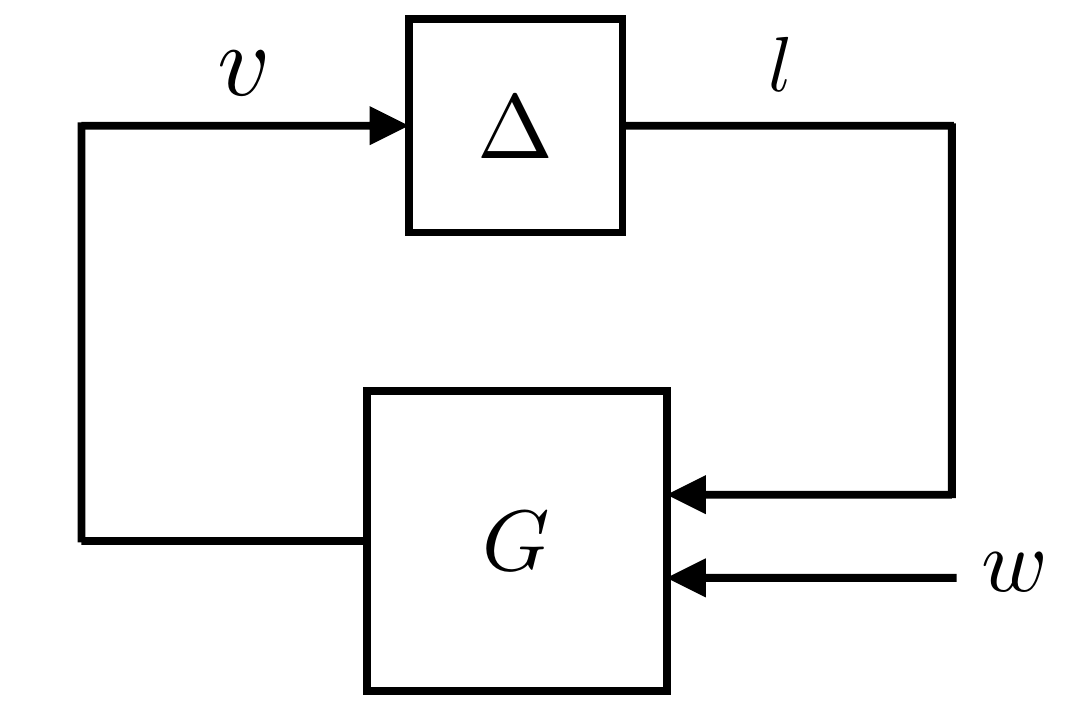}
	\caption{Interconnection $F_u(G,\Delta)$ of a nominal nonlinear system $G$ and a perturbation $\Delta$}
	\label{fig:Fu}    
\end{figure}

From robust control modeling \cite{Zhou:1996}, the perturbation $\Delta$ can represent various types of nonlinearity and uncertainty, including hard nonlinearities (e.g. saturation), time delays, and unmodeled dynamics. Different types of perturbation have different input-output properties, and each property can be described by its corresponding IQCs~\cite{Megretski:97}.  To help define IQCs, we introduced a virtual filter $\Psi$ (shown in Fig~\ref{fig:filter}) that is an linear time invariant (LTI) system, driven by the input $v$ and output $l$ of $\Delta$, and with zero initial condition $x_{\psi}(t_0) = 0^{n_\psi}$. Its dynamics are given by
\begin{subequations} \label{eq:filter_sys}
\begin{align}
\dot{x}_{\psi}(t) =& A_{\psi} x_{\psi}(t) + B_{\psi 1}v(t) + B_{\psi 2}l(t),  \\
z(t) =& C_{\psi} x_{\psi}(t) + D_{\psi 1}v(t) + D_{\psi 2}l(t), \label{eq:z_def}
\end{align}
\end{subequations}
where $x_{\psi}(t) \in \mathbb{R}^{n_{\psi}}$ is the state, and $z(t) \in \mathbb{R}^{n_z}$ is the output. 
For many types of perturbations (e.g. the ones in Example~\ref{ex:IQC_example_LTI} -- \ref{ex:IQC_example_para}), we can choose $\Psi$ to be an identity matrix, i.e., $z = [v;l]$. But dynamic filters are able to capture the correlation between the input and output signals of $\Delta$ across time, which enriches the description of $\Delta$. For examples on dynamic filters, the reader is referred to \cite{Zames:68, Megretski:97,Veenman:16}. IQCs can be either defined in frequency domain or time domain. The use of time domain IQCs is required by the dissipation inequality used in the paper.  Time domain IQCs consist of soft IQCs and hard IQCs, which are quadratic constraints on the output $z$ associated with a matrix $M$ over infinite (soft IQC) or finite (hard IQC) horizons. The definition for hard IQCs is given below, the use of soft IQCs is discussed in Section~\ref{sec:soft_IQC}.

\begin{figure}[h]
	\centering
	\includegraphics[width=0.4\textwidth]{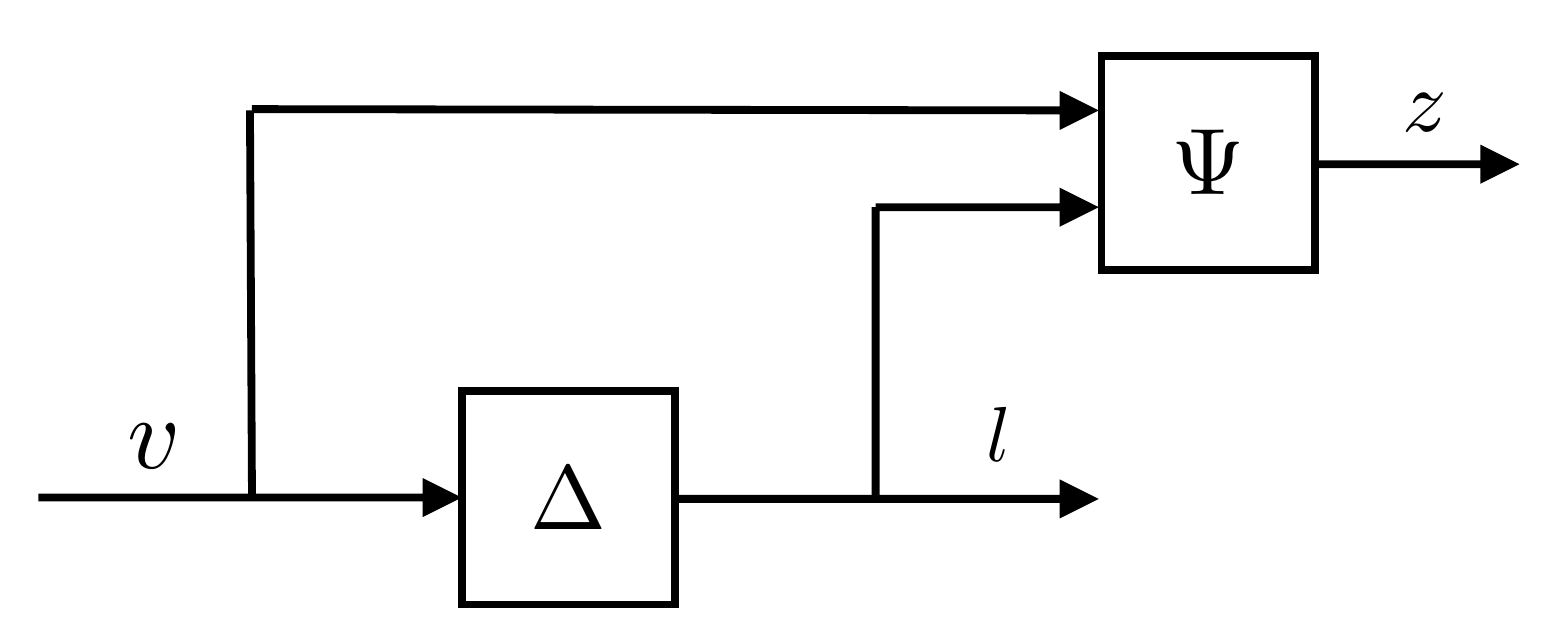}
	\caption{Graphical interpretation for time domain IQCs}
	\label{fig:filter}    
\end{figure}

\begin{definition}
	Let $\Psi \in \mathbb{RH}_{\infty}^{n_z \times (n_v + n_l)}$ and $M \in \mathbb{S}^{n_z}$ be given.  A bounded, causal operator  $\Delta: \mathcal{L}_{2e}^{n_v} \rightarrow \mathcal{L}_{2e}^{n_l}$ satisfies the \underline{hard IQC} defined by $(\Psi, M)$ if the following condition holds for all $v \in \mathcal{L}_{2e}^{n_v}, \ and \ l = \Delta(v)$:
		\begin{align}
		\int_{t_0}^t z(\tau)^\top M z(\tau) d\tau \ge 0, \ \forall t \in [t_0, T], \label{eq:hardIQC}
		\end{align}
	where $z = \Psi \left[ \begin{smallmatrix}v \\l \end{smallmatrix}\right]$ (Eq. \ref{eq:z_def}) is the output of $\Psi$ driven by the inputs $(v,l)$.
\end{definition}
 We use the notation $\Delta \in$ HardIQC$(\Psi, M)$ to indicate that $\Delta$ satisfies the hard IQC specified ($\Psi, M$), \revise{i.e., given any input $v$ of $\Delta$, the output $l$ must be such that $z = \Psi \left[\begin{smallmatrix} v \\ l\end{smallmatrix}\right]$ satisfies the constraint \eqref{eq:hardIQC} characterized by $(\Psi, M)$.} Next, we give two examples on different types of uncertainties and the corresponding hard IQCs.
\begin{example}\label{ex:IQC_example_LTI}
Consider the set $\mathcal{S}_1$ of LTI uncertainties with a given norm bound $\sigma >0$, i.e., $\Delta \in \mathcal{S}_1$, if $\Delta \in \mathbb{RH}_{\infty}$ with $\norm{\Delta}_{\infty} \leq \sigma$. 
It's proved in  \cite{Balakrishnan:02} that $\Delta \in HardIQC(\Psi, M_{D})$ over any finite horizon $T < \infty$, where $\Psi := \left[\begin{smallmatrix}\Psi_{11} \ \ &0 \\ 0 \ \ &\Psi_{11}\end{smallmatrix}\right]$ with $\Psi_{11} \in \mathbb{RH}_{\infty}^{n_z\times 1}$ and
\begin{align} \label{eq:MD}
\revise{M_D \in \mathcal{M}_1 := \left\{ \left[\begin{smallmatrix} \sigma^2 M_{11} \ \ & 0 \\ 0\ \ & -M_{11}\end{smallmatrix}\right] : M_{11} \succeq 0 \right\}.}
\end{align}
A typical choice for $\Psi_{11}$ \cite{Veenman:16} is
\begin{align}
	\Psi_{11}^{d,m} = \bmat{1, \frac{1}{(s+m)}, \cdots, \frac{1}{(s+m)^d}}^\top, \text{with} \ \ m > 0, \label{eq:filter_choice}
\end{align}
where $m$ and $d$ are selected by the user.
\end{example}

\begin{example}\label{ex:nonlinear}
Consider the set $\mathcal{S}_2$ of nonlinear, time varying, uncertainties with a given norm-bound $\sigma$, i.e. $\Delta \in \mathcal{S}_2$, if $\norm{\Delta}_{{2 \rightarrow 2},[t_0, T]} \leq \sigma$. If $\Delta \in \mathcal{S}_2$, then the perturbation $\Delta$ satisfies the hard IQCs defined by $(\Psi, M)$ defined below:
\begin{align}
\Psi = I_{n_{v} + n_l}, \ \revise{M \in \mathcal{M}_2 := \left\{\left[\begin{smallmatrix}\sigma^2 \lambda I_{n_v} & 0\\ 0 & -\lambda I_{n_l} \end{smallmatrix}\right]: \  \lambda \ge 0 \right\}.}
\end{align}
\end{example}

Since the behavior of the perturbation $\Delta$ can be described by an IQC associated with a filter $\Psi$ and a matrix $M$, then the robust analysis on the original uncertain system $F_u(G,\Delta)$ can be instead conducted on the extended system shown in Fig.~\ref{fig:G_Psi} with an additional constraint \eqref{eq:hardIQC}. \revise{The precise relation $l = \Delta(v)$, for analysis, is replaced by the constraint on $z$.} This extended system is an interconnection of $G$ and $\Psi$, with $\Delta$ been removed. The dynamics of the extended system are of the form:
\begin{subequations}\label{eq:extended_sys}
	\begin{align}
	\dot{x}(t) &= F(t,x(t),l(t),w(t),\delta(t)), \\
	z(t) &= H(t,x(t),l(t),w(t),\delta(t)), \label{eq:output_H}
	\end{align}
\end{subequations}
where the state $x := [x_G; x_\psi] \in \mathbb{R}^n, n = n_G + n_\psi$, gathers the state of $G$ and $\Psi$. \revise{The mappings $F$, and $H$ are given by (dropping the dependence on $t$):
\begin{equation} \label{eq:FH_def}
\begin{array}{lll}
F(t,x,l,w,\delta) := &\bmat{f(t, x_G, l, w, \delta) \\ A_{\psi} x_{\psi} + B_{\psi 1}r(t, x_G, l, w, \delta) + B_{\psi 2}l },  \\
H(t,x,l,w,\delta) := &C_{\psi} x_{\psi} + D_{\psi 1}r(t, x_G, l, w, \delta) + D_{\psi 2}l.
\end{array}
\end{equation} }
\begin{figure}[h]
	\centering
	\includegraphics[width=0.42\textwidth]{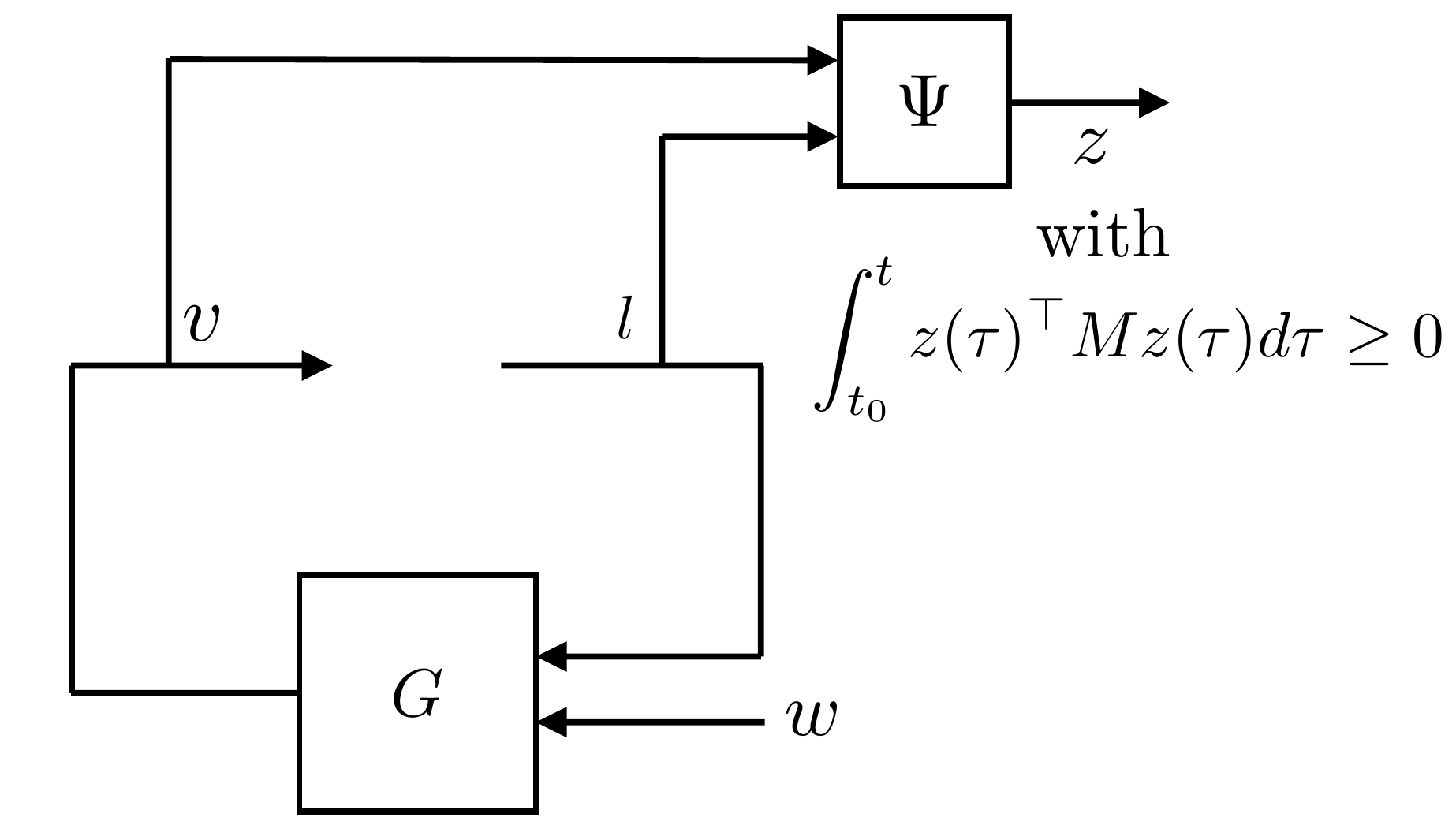}
	\caption{Extended system of $G$ and $\Psi$}
	\label{fig:G_Psi}    
\end{figure}

The original uncertain system to be analyzed is $F_u(G,\Delta)$, which has a set of initial conditions $\mathcal{X}_0$ and an input $w$. The analysis is instead conducted on the extended system \eqref{eq:extended_sys}, which has a set of initial conditions $\mathcal{X}_0\times \{0^{n_\psi}\}$, and two inputs $w$ and $l$. \revise{For any input $w \in \mathcal{L}_{2}^{n_w}$ and initial condition $x_G(t_0) \in \R^{n_G}$, the solutions $v \in \mathcal{L}_{2e}^{n_v}$ and $l \in \mathcal{L}_{2e}^{n_l}$ to the original system $F_u(G,\Delta)$ satisfy the constraint \eqref{eq:hardIQC}. The extended system \eqref{eq:extended_sys} with the IQC \eqref{eq:hardIQC} ``covers" the responses of the original uncertain system $F_u(G,\Delta)$. Specifically, given any input $w \in \mathcal{L}_{2}^{n_w}$ and initial condition $x_G(t_0) \in \R^{n_G}$, the input $l \in \mathcal{L}_{2e}^{n_l}$ is implicitly constrained in the extended system so that the pair $(v,l)$ satisfies the IQC \eqref{eq:hardIQC}. This set of $(v,l)$ that satisfies the IQC \eqref{eq:hardIQC} includes all input/output pairs of $\Delta$. Therefore, the response of this extended system subject to this implicit constraint \eqref{eq:hardIQC} includes all behaviors of the original uncertain system $F_u(G,\Delta)$. The following theorem provides the method for outer bounding the FRS of the uncertain system $F_u(G,\Delta)$ by conducting analysis on the constrained extended system \eqref{eq:extended_sys}}.
\begin{theorem} \label{thm:hardIQCs}
	\revise{Let $G$ be a nonlinear system defined by \eqref{eq:sysG4IQC},  and $\Delta: \mathcal{L}^{n_v}_{2e} \rightarrow \mathcal{L}^{n_l}_{2e}$ be a bounded and causal operator. Let Assumption~\ref{ass:ass1} hold.} Additionally, assume \revise{(i) $F_u(G,\Delta)$ is well-posed,} (ii) $\Delta \in$ HardIQC$(\Psi, M)$, with $\Psi$ and $M$ given, and (iii) all the trajectories of the extended system start from $\mathcal{X}_0\times \{0^{n_\psi}\}$. For some \revise{$F$, $H$ defined in \eqref{eq:FH_def}}, time interval $[t_0, T]$, local region $\mathcal{X}_l \subset \mathbb{R}^{n_G}$, set of initial conditions $\mathcal{X}_0 \subset \mathbb{R}^{n_G}$, disturbance bound $R$, function $h$, and set of uncertain parameters $\mathcal{D}$, function $q: \R^{n_G} \rightarrow \R$, and $\alpha \in \R$, suppose there exists a $\mathcal{C}^1$ function $V: \mathbb{R} \times \mathbb{R}^{n} \rightarrow \mathbb{R}$ that satisfies
	\begingroup
	\allowdisplaybreaks
	\begin{subequations} \label{eq:cond_hardIQC}
	\begin{align}
	&\frac{\partial V(t,x)}{\partial t} + \frac{\partial V(t,x)}{\partial x}F(t,x,l,w,\delta)  + z^\top M z \leq w^\top w, \ \forall (x,t,l,w,\delta) \in  \mathcal{X}_l  \times \mathbb{R}^{n_{\psi}} \times [t_0, T] \times \mathbb{R}^{n_l} \times \mathbb{R}^{n_w} \times \mathcal{D}, \label{eq:constraintF1} \\
	&\mathcal{X}_0 \times \{0^{n_\psi}\} \subseteq \{x \in \mathbb{R}^{n} : V(t_0, x) \leq 0\}, \label{eq:constraintF2} \\
	&\left\{x_G \in \mathbb{R}^{n_G} : V(T, x) \leq R^2\right\} \subseteq \Omega_{\alpha}^{q}, \ \forall x_{\psi} \in \mathbb{R}^{n_{\psi}},  \label{eq:constraintF3} \\
	&\left\{x_G \in \mathbb{R}^{n_G} : V(t, x) \le R^2 h(t) \right\} \subseteq \mathcal{X}_l , \ \forall (t,x_\psi) \in [t_0, T] \times \mathbb{R}^{n_{\psi}},  \label{eq:constraintF4}
	\end{align}
	\end{subequations}
	\endgroup
	where \revise{$z$ is the output of the map $H$. Then all the trajectories of $F_u(G,\Delta)$ (defined by \eqref{eq:sysG4IQC} -- \eqref{eq:def_Delta}) starting from $x_G(t_0) \in \mathcal{X}_0$ satisfy $x_G(T) \in \Omega_{\alpha}^{q}$.}  Therefore $\Omega_{\alpha}^q$ is an outer bound to the $FRS(T;F_u(G,\Delta),t_0,\mathcal{X}_0,R,h,\mathcal{D})$ \eqref{def:FRS_Fu}.
\end{theorem}
\begin{proof}
By assumption that $F_u(G,\Delta)$ is well-posed, the signals $(x,v,l,z)$ generated for the extended system for the input $w \in \mathcal{L}_2^{n_w}$ are $\mathcal{L}_{2e}$ signals. By combining \eqref{eq:constraintF1} and \eqref{eq:constraintF4} we have the following dissipation inequality:
\begin{align}
	&\frac{\partial V(t,x)}{\partial t} + \frac{\partial V(t,x)}{\partial x}F(t,x,l,w,\delta) + z^\top M z \leq w^\top w, \nonumber \\
	&\quad\quad \quad \quad\quad\quad\quad\quad\quad\quad\quad\quad\quad\quad\forall (x, t,l,w, \delta) \ \text{s.t.} \ x \in \Omega^V_{t,R^2h(t)}, t \in [t_0, T], l \in \mathbb{R}^{n_l}, w \in \mathbb{R}^{n_w}, \delta \in \mathcal{D}. \label{eq:diss_hardIQC}
\end{align}
Since \eqref{eq:diss_hardIQC} only holds on the set $\Omega_{t,R^2h(t)}^V$, we need to first prove that all the states starting from $\mathcal{X}_0 \times \{0^{n_\psi}\}$ won't leave $\Omega_{t,R^2h(t)}^V$, for all $t \in [t_0, T]$. Assume there exist a time instance $T_1 \in [t_0,T]$, $x_0 \in \mathcal{X}_0 \times \{0^{n_\psi}\}$, and signals $w$ satisfying \eqref{eq:w_rate}, $\delta(t) \in \mathcal{D}$, $l(t) \in \R^{n_l}$, such that a trajectory of the extended system starting from $x(t_0) = x_0$ satisfies $V(T_1,x(T_1)) > R^2 h(T_1)$. Define $T_2 = \inf_{V(t,x(t))>R^2 h(t)} t$, and integrate \eqref{eq:diss_hardIQC} over $[t_0,T_2]$:
\begin{align}
V(T_2,x(T_2)) - V(t_0,x(t_0)) + \int_{t_0}^{T_2} z(t)^\top M z(t) dt &\leq  \int_{t_0}^{T_2} w(t)^\top w(t) dt. \nonumber  \\
\intertext{\revise{By assumnption $x_0 \in \mathcal{X}_0 \times \{0^{n_\psi}\}$, it follows from constraint \eqref{eq:constraintF2} that $V(t_0,x(t_0)) \leq 0$. Comibining it with $w$ satisfing \eqref{eq:w_rate} yields}}
V(T_2,x(T_2)) + \int_{t_0}^{T_2} z(t)^\top M z(t) dt & <  R^2h(T_2). \label{eq:first_ineq} \\
\intertext{\revise{Next it follows from the hypothesis that $\Delta \in $ HardIQC$(\Psi,M)$ that }}
R^2 h(T_2) = V(T_2,x(T_2)) & < R^2h(T_2). \label{eq:second_ineq}
\end{align}
We can see the contradiction in \eqref{eq:second_ineq}. Therefore there doesn't exist a $T_1 \in [t_0,T]$, such that $x(T_1) \notin \Omega_{T_1,R^2 h(T_1)}^V$. As a result, for all $x(t_0) \in \mathcal{X}_0 \times \{0^{n_\psi}\}$, we have $x(t) \in \Omega_{t,R^2h(t)}^V$, for all $t \in [t_0, T]$, and thus $x(T) \in \Omega_{T,R^2}^V$. Finally, it follows from \eqref{eq:constraintF3} that $x_G(T) \in \Omega_{\alpha}^q$. 
\end{proof}

Notice that from the proof, $\Omega_{T,R^2}^V$ is an outer bound to the FRS of the extended system from $\mathcal{X}_0 \times \{0^{n_\psi}\}$. The set $\Omega_{\alpha}^q$, a projection of $\Omega_{T,R^2}^V$ on the $x_G$ space, is an outer bound to the FRS of the actual uncertain system $F_u(G,\Delta)$.

There is a large library of IQCs for various types of perturbations $\Delta$ \cite{Megretski:97}. It is common to formulate optimization problems that search over combinations of valid IQCs.  Specifically, let $\left\{ (\Psi_k,M_k) \right\}_{k=1}^N$ be a collection of valid time-domain IQCs for a particular $\Delta$. If $z_k$ is the output of the filter $\Psi_k$ and $\lambda_1,....,\lambda_N$ are non-negative scalars then it follows that:
\begin{align}
&\int_{t_0}^T \sum_{k=1}^N \lambda_{k} z_k(t)^\top M_k(t) z_k(t) dt \ge 0, \ \forall v_k \in \mathcal{L}_{2e}^{n_{v_k}}, \ \ l_k = \Delta(v_k), \ \ \text{and} \ \ T \ge t_0. \nonumber 
\end{align}
In other words, a conic combination of time-domain IQCs is also an IQC.  This conic combination can be represented as $\Psi := [\Psi_1; ...; \Psi_N]$ and $M := blkdiag(\lambda_1 M_1,...,\lambda_N M_N)$. The scalars $\lambda_1,...\lambda_N \ge 0$ are typically decision variables in an optimization used to find the best IQC for the robustness analysis. In this parameterization $\Psi$ is fixed and $M$ is a linear function of variables $\lambda_1,...,\lambda_N$ subject to non-negativity constraints. More general IQC parameterizations can be found in \cite{Veenman:16}: given the type of the perturbation, the corresponding IQCs are parametrized by a fixed filter $\Psi$ chosen by the analyst and $M$ in a feasible set $\mathcal{M}$ described by linear matrix inequality (LMI) constraints. These general parametrizations will be used in the rest of the paper. \revise{Note that Example~\ref{ex:IQC_example_LTI} and \ref{ex:nonlinear} also provide instances of the general parametrization, where $M_D$ and $M$ are restricted to convex sets $\mathcal{M}_1$ and $\mathcal{M}_2$.}

\revise{Along with $V$, we also treat $M \in \mathcal{M}$ as a decision variable to give the optimization more flexibility. Assume the set $\mathcal{M}$ is convex and described by LMIs.} 
Again, assume $\mathcal{X}_l$, $\mathcal{X}_0$ are parametrized by $p \in \R[x_G]$ and $r_0 \in \R[x_G]$, respectively, and restrict $q \in \R[x_G]$. \revise{By applying the generalized S-procedure \cite{Parrilo:00} to \eqref{eq:cond_hardIQC}, we obtain} the following SOS optimization problem, $\boldsymbol{sosopt_2(F, H, p, g, q, r_0, R, h, p_\delta, \Psi, \mathcal{M})}$, 
\begingroup
\allowdisplaybreaks
\begin{subequations}
\begin{align} 
\min_{\alpha, \eta, s, V, M, \epsilon_1, \epsilon_2} & \ \ \alpha \nonumber \\
\text{s.t.} \ \ \ \ \ 
&s_5 - \epsilon_1 \in \Sigma[x], s_6 - \epsilon_2 \in \Sigma[(x,t)],\epsilon_1 > 0, \epsilon_2 > 0, \ M \in \mathcal{M}, \ V \in \mathbb{R}[(t,x)], \nonumber \\
&s_4 \in \Sigma[x_G], \ s_7 \in \Sigma[(x,t)], \ s_i \in \Sigma[(x,w,l,t,\delta)], \ \forall i \in \{1,2,3\},  \nonumber \\
& -\bigg(\frac{\partial  V}{\partial t} + \frac{\partial  V}{\partial x} F + z^\top Mz- w^\top w \bigg)   + \revise{(p - \eta)}s_1 -s_2g - s_3 p_\delta  \in \Sigma[(x,w,l,t,\delta)],  \\
& -V\vert_{t = t_0, x=[x_G;0^{n_\psi}]} + s_4 r_0 \in \Sigma[x_G],  \\
& -( q - \alpha)s_5 + V\vert_{t = T} - R^2 \in \Sigma[x],  \\
& - \revise{(p-\eta)}s_6 + V - R^2h  - s_7 g \in \Sigma[(x,t)], 
\end{align}
\end{subequations}
\endgroup
which \revise{is again bilinear in $(\alpha, \eta)$ and $(s_1, s_5, s_6)$, and can be solved by using Algorithm~\ref{alg:sol_algo}}. Although in the SOS formulation, $M$ is restricted to be time-invariant, extensions to allow for time-varying $M$ are possible. 

\revise{To keep track of all the tuning parameters in the paper, we provide a table that summarizes them, their corresponding physical meanings, and some of their examples:}
\begin{table}[h]
	\centering
	\caption{\revise{List of tuning parameters} \label{tab:table_tuning}}
	\begin{tabular}{ |M{2.9cm}|M{2.7cm}|M{3cm}|M{3.4cm}|M{3cm}|}
		\hline
		Physical meanings &Shape of $\mathcal{X}_l$ &Outer bound shape &Energy releasing rate &Filter for $\Delta$ \\
		\hline
		Parameters &$p$ & $q$ &$h$  & $\Psi$ \\ \hline 
		Examples &Sections~\ref{ex1}, \ref{sec:softhard_compare}  & Sections~\ref{ex1}, \ref{sec:softhard_compare} &Section~\ref{ex1}  & Sections~\ref{sec:softhard_compare}, \ref{ex6} \\ \hline 
	\end{tabular}
\end{table}

\section{Robust Reachability Analysis  with Soft IQCs}\label{sec:soft_IQC}
The previous section gives the result using hard IQCs, however, the library of IQCs are usually provided in frequency domain \cite{Megretski:97}, whose definition is given below:
\begin{definition}
	Let $\Pi = \Pi^\sim \in \mathbb{RL}_{\infty}^{(n_v + n_l)\times (n_v + n_l)}$ be given. A bounded, causal operator $\Delta: \mathcal{L}^{n_v}_{2e} \rightarrow  \mathcal{L}^{n_l}_{2e}$ satisfies the \underline{frequency domain IQC} defined by the multiplier $\Pi$, if the following inequality holds for all $v \in \mathcal{L}_{2}^{n_v}$ and $l = \Delta(v)$,
	\begin{align}
	\bigintss_{-\infty}^{\infty} \bmat{\hat{v}(j\omega)\\\hat{l}(j\omega)}^* \Pi(j\omega) \bmat{\hat v (j\omega)\\\hat l(j\omega)} d\omega \ge 0, \label{eq:freqQC}
	\end{align}
	where $\hat v$ and $\hat l$ are Fourier transforms of $v$ and $l$.
\end{definition} 
\revise{The frequency domain multiplier can be factorized as $\Pi = \Psi^\sim M \Psi$ where $M \in \mathbb{S}^{n_z}$ and $\Psi$ is a stable, LTI system of appropriate dimension. Such a factorization always exists \cite{Veenman:16} but is not unique.  This factorization $(\Psi,M)$ gives rise to a time-domain soft IQC as defined next.} 
\begin{definition}
	Let $\Psi \in \mathbb{RH}_{\infty}^{n_z \times (n_v + n_l)}$ and $M \in \mathbb{S}^{n_z}$ be given. A bounded, causal operator  $\Delta: \mathcal{L}_{2e}^{n_v} \rightarrow \mathcal{L}_{2e}^{n_l}$ satisfies the \underline{soft IQC} defined by $(\Psi, M)$ if the following inequality holds for all $v \in \mathcal{L}_{2}^{n_v} \ and \ l = \Delta(v)$:
	\begin{align}
	\int_{t_0}^{\infty} z(t)^\top M z(t) dt \ge 0, \label{eq:softIQC}
	\end{align}
	where $z = \Psi \left[ \begin{smallmatrix}v \\l \end{smallmatrix}\right]$ (Eq. \ref{eq:z_def}) is the output of $\Psi$ driven by the inputs $(v,l)$.
\end{definition}
We use the notation $\Delta \in$ IQC$(\Pi)$ and $\Delta \in $ SoftIQC$(\Psi, M)$ to indicate that $\Delta$ satisfies the corresponding frequency domain and soft IQC, \revise{meaning that given any $v$, the output $l$ of $\Delta$ must be such that \eqref{eq:freqQC} and \eqref{eq:softIQC} hold, respectively}. By Parseval's theorem \cite{Zhou:1996}, frequency domain and time domain soft IQCs are equivalent. \revise{Specifically, if $\Delta \in$ IQC$(\Pi)$ then $\Delta \in$ SoftIQC$(\Psi,M)$ for any factorization $\Pi = \Psi^\sim M \Psi$ with $\Psi$ stable.  Conversely if $\Delta \in$ SoftIQC$(\Psi,M)$ then $\Delta \in$ IQC$(\Psi^\sim M \Psi)$ as well. It also follows that $\Delta \in$ HardIQC$(\Psi,M)$ implies $\Delta \in$ IQC$(\Psi^\sim M \Psi)$.  However, $\Delta \in$ IQC$(\Pi)$ does not imply, for general factorizations, that $\Delta \in$ HardIQC$(\Psi,M)$.} As a result, soft IQCs are always available while hard ones are not, which necessitates the use of soft IQCs in the dissipation inequality. Next, we give one example of uncertainty and its corresponding soft IQC.

\begin{example}\label{ex:IQC_example_para}
	Consider the set $\mathcal{S}_3$ of real constant parametric uncertainties with given norm bound $\sigma >0$, i.e. $\Delta \in \mathcal{S}_3$, if $l(t) = \Delta(v(t)) = \delta_{TI} v(t)$ with $ \vert \delta_{TI} \vert \leq \sigma$. From \cite{Megretski:97}, the frequency domain filter is chosen as $\Pi_\delta = \left[\begin{smallmatrix}\sigma^2\Pi_{11}(j\omega) & \Pi_{12}(j\omega) \\ \Pi_{12}^*(j\omega) & -\Pi_{11}(j\omega) \end{smallmatrix} \right]$, where $\Pi_{11}(j\omega)=\Pi_{11}^*(j\omega) \ge 0$ and $\Pi_{12}(j\omega)=-\Pi_{12}^*(j\omega)$ for all $\omega$. A soft IQC factorization for $\Pi_\delta$ is $\Psi= \left[\begin{smallmatrix} \Psi_{11}^{d,m} & 0 \\ 0 & \Psi_{11}^{d,m}\end{smallmatrix}\right]$, where $\Psi_{11}^{d,m}$ is defined in \eqref{eq:filter_choice}, and $M_{DG} = \left[\begin{smallmatrix} \sigma^2 M_{11} & M_{12} \\ M_{12}^\top & -M_{11} \end{smallmatrix} \right]$, where decision matrices are subject to $M_{11} = M_{11}^\top$, $M_{12} = -M_{12}^\top$ and $\Psi^{d,m\sim}_{11} M_{11}\Psi_{11}^{d,m} \ge 0$. The constraints $\Psi^{d,m\sim}_{11} M_{11}\Psi_{11}^{d,m} \ge 0$ can be enforced by a KYP LMI \cite{RANTZER19967}. Notice that $\delta_{TI}$ is a special case of the perturbation considered in Example \ref{ex:IQC_example_LTI}, and  thus $\delta_{TI} \in$ HardIQC$(\Psi, M_D)$ as well. However, since $M_D$ is a special case of $M_{DG}$ with $M_{12} \equiv 0$, the reachability analysis using $(\Psi, M_{DG})$ can be less conservative than using $(\Psi, M_D)$. 
\end{example}

 Soft IQCs are constraints that hold over the infinite time horizon and hence they cannot be directly incorporated in the analysis based on finite-horizon dissipation inequalities. The following Lemma is a remedy for this issue: it provides a lower-bound for soft IQCs on finite horizons then enabling their use for reachability analysis. This in turn enables us to: (i) conduct reachability analysis when the hard IQC factorization does not exist; (ii) reduce conservatism resulting from the hard IQC factorization when it exists, as discussed in Example \ref{ex:IQC_example_para} . 
\begin{lemma} \label{lemma:iqc_bound} (\cite{Fetzer:2018}) Let $\Psi \in \mathbb{RH}_{\infty}^{n_z \times (n_v+n_l)}$ and $M \in \mathbb{S}^{n_z}$ be given. Define $\Pi := \Psi^{\sim}M \Psi$. If $\Pi_{22}(j\omega) < 0 \ \forall \omega$, then\footnote{The notation $\Pi_{22}$ refers to the partitioning $\Pi = \left[ \begin{smallmatrix}\Pi_{11} & \Pi_{12} \\ \Pi_{12}^\sim & \Pi_{22}\end{smallmatrix} \right]$ conformably with the dimensions of $v$ and $l$.}
	\begin{itemize}
		\item $D_{\psi2}^\top M D_{\psi2} < 0$ and there exists a  $Y_{22} \in \mathbb{S}^{n_\psi}$ satisfying
		 \begin{align}\label{eq:cond_Y22}
		 \revise{KYP(Y_{22}, A_\psi, B_{\psi2}, C_\psi, D_{\psi2},M) < 0.}
		 \end{align}
		\item If $\Delta \in SoftIQC(\Psi, M)$ then for all $T \ge 0$, $v \in \mathcal{L}_{2e}^{n_v}$ and $l = \Delta(v)$,
		\begin{align}
		\int_0^T z(t)^\top M z(t) dt \ge - x_\psi(T)^\top Y_{22} x_\psi(T) \label{eq:IQC_lowerbound}
		\end{align}
		for any $Y_{22} \in \mathbb{S}^{n_\psi}$ satisfying \eqref{eq:cond_Y22}.
	\end{itemize}
\end{lemma}

\revise{Lemma~\ref{lemma:iqc_bound} is valid for multipliers that satisfy $\Pi_{22} >0$. Multipliers satisfying the non-strict conditions $\Pi_{22} \ge 0$ can be handled by a perturbation argument \cite{Pete_IQC}.} Based on the lemma given above, the following theorem considers the analysis for the interconnection $F_u(G,\Delta)$ with $\Delta$ that has a soft IQC factorization.
\begin{theorem} \label{thm:soft_IQC}
	\revise{Let $G$ be a nonlinear system defined by \eqref{eq:sysG4IQC},  and $\Delta: \mathcal{L}^{n_v}_{2e} \rightarrow \mathcal{L}^{n_l}_{2e}$ be a bounded and causal operator. Let Assumption~\ref{ass:ass1} hold.} Additionally, assume \revise{(i) $F_u(G,\Delta)$ is well-posed,} (ii) $\Delta \in$ SoftIQC$(\Psi, M)$, with $\Psi$ and $M$ given, \revise{(iii) $\Pi:= \Psi^\sim M \Psi$ satisfying $\Pi_{22} <0 \ \forall \omega$,} and (iv) all the trajectories of the extended system start from $\mathcal{X}_0 \times \{0^{n_\psi}\}$. For some \revise{$F$, $H$ defined in \eqref{eq:FH_def}}, time interval $[t_0, T]$, local region $\mathcal{X}_l \subset \mathbb{R}^{n_G}$, set of initial conditions $\mathcal{X}_0 \subset \mathbb{R}^{n_G}$, disturbance bound $R$, function $h$, and set of uncertain parameters $\mathcal{D}$, function $q:\R^{n_G} \rightarrow \R$, and $\alpha \in \R$, suppose there exists a $\mathcal{C}^1$ function $V : \mathbb{R} \times \mathbb{R}^{n} \rightarrow \mathbb{R}$, and a matrix $Y_{22} \in \mathbb{S}^{n_{\psi}}$ satisfying \eqref{eq:cond_Y22}, such that the following constraints hold
	\begingroup
	\allowdisplaybreaks
	\begin{subequations} \label{eq:cond_softIQC}
	\begin{align}
	&\frac{\partial V(t,x)}{\partial t} + \frac{\partial V(t,x)}{\partial x}F(t,x,l,w,\delta) + z^\top M z \leq w^\top w, \ \forall (x, t,l,w,\delta) \in  \mathcal{X}_l  \times \mathbb{R}^{n_{\psi}} \times [t_0, T] \times \mathbb{R}^{n_l} \times \mathbb{R}^{n_w} \times \mathcal{D}, \label{eq:constraintG1} \\
	&\mathcal{X}_0\times \{0^{n_\psi}\} \subseteq \{x \in \mathbb{R}^n : V(t_0, x) \leq 0\}, \label{eq:constraintG2} \\
	&\left\{x_G \in \mathbb{R}^{n_G} : \revise{\mathcal{V}(T, x)}  \leq R^2\right\} \subseteq \Omega_{\alpha}^{q}, \forall x_{\psi} \in \mathbb{R}^{n_{\psi}}, \label{eq:constraintG3} \\
	&\left\{x_G \in \mathbb{R}^{n_G} : \revise{\mathcal{V}(t,x)}  \le R^2 h(t) \right\} \subseteq \mathcal{X}_l , \forall (t,x_\psi) \in [t_0, T] \times \mathbb{R}^{n_{\psi}}.  \label{eq:constraintG4} 
	\end{align}
	\end{subequations}
	\endgroup
	where \revise{$\mathcal{V} = V - x_\psi^\top Y_{22} x_\psi$, and $z$ is the output of the map $H$. Then all the trajectories of $F_u(G,\Delta)$ (defined by \eqref{eq:sysG4IQC} -- \eqref{eq:def_Delta}) starting from $x_G(t_0) \in \mathcal{X}_0$ satisfy $x_G(T) \in \Omega_{\alpha}^{q}$.} Therefore $\Omega_{\alpha}^q$ is an outer bound to $FRS(T;F_u(G,\Delta),t_0,\mathcal{X}_0,R,h,\mathcal{D})$ \eqref{def:FRS_Fu}.
\end{theorem}
\revise{
\begin{proof}
	By assumption that $F_u(G,\Delta)$ is well-posed, the signals $(x,v,l,z)$ generated for the extended system for the input $w \in \mathcal{L}_2^{n_w}$ are $\mathcal{L}_{2e}$ signals. By combining \eqref{eq:constraintG1} and \eqref{eq:constraintG4} we have the following dissipation inequality:
	\begin{align}
	&\frac{\partial V(t,x)}{\partial t} + \frac{\partial V(t,x)}{\partial x}F(t,x,l,w,\delta) + z^\top M z \leq w^\top w, \nonumber \\
	&\quad \quad\quad\quad \quad \quad\quad\quad\quad\quad\quad\quad\quad\quad\forall (x, t,l,w, \delta) \ \text{s.t.} \ x \in \Omega^{\mathcal{V}}_{t,R^2h(t)}, t \in [t_0, T], l \in \mathbb{R}^{n_l}, w \in \mathbb{R}^{n_w}, \delta \in \mathcal{D}. \label{eq:diss_softIQC}
	\end{align}
	Since \eqref{eq:diss_softIQC} only holds on the set $\Omega_{t,R^2h(t)}^{\mathcal{V}}$, we need to first prove that all the states starting from $\mathcal{X}_0 \times \{0^{n_\psi}\}$ won't leave $\Omega_{t,R^2h(t)}^{\mathcal{V}}$, for all $t \in [t_0, T]$. Assume there exist a time instance $T_1 \in [t_0,T]$, $x_0 \in \mathcal{X}_0 \times \{0^{n_\psi}\}$, and signals $w$ satisfying \eqref{eq:w_rate}, $\delta(t) \in \mathcal{D}$, $l(t) \in \R^{n_l}$, such that a trajectory of the extended system starting from $x(t_0) = x_0$ satisfies ${\mathcal{V}}(T_1,x(T_1)) > R^2 h(T_1)$. Define $T_2 = \inf_{\mathcal{V}(t,x(t))>R^2 h(t)} t$, and integrate \eqref{eq:diss_softIQC} over $[t_0,T_2]$:
	\begin{align}
	V(T_2,x(T_2)) - V(t_0,x(t_0)) + \int_{t_0}^{T_2} z(t)^\top M z(t) dt &\leq  \int_{t_0}^{T_2} w(t)^\top w(t) dt. \nonumber  \\
	\intertext{By assumption $x_0 \in \mathcal{X}_0 \times \{0^{n_\psi}\}$, it follows from constraint \eqref{eq:constraintG2} that $V(t_0,x(t_0))\leq 0$. Combining it with $w$ satisfying \eqref{eq:w_rate} to show}
	V(T_2,x(T_2)) + \int_{t_0}^{T_2} z(t)^\top M z(t) dt & <  R^2h(T_2). \\
	\intertext{It follows from Lemma~\ref{lemma:iqc_bound}, $\Delta \in $ SoftIQC$(\Psi,M)$, $\Pi_{22} < 0 \ \forall \omega$ and $Y_{22}$ satisfies \eqref{eq:cond_Y22} that \eqref{eq:IQC_lowerbound} holds for the $\mathcal{L}_{2e}$ signals $(v,l,z)$, and thus}
	 V(T_2,x(T_2)) - x_\psi(T_2)^\top Y_{22} x_\psi(T_2) & < R^2h(T_2). \label{eq:contradict}
	\end{align}
	Thus \eqref{eq:contradict} is a contradiction, since $V(T_2,x(T_2)) - x_\psi(T_2)^\top Y_{22} x_\psi(T_2) = \mathcal{V}(T_2,x(T_2)) = R^2h(T_2)$. Therefore there doesn't exist a $T_1 \in [t_0,T]$, such that $x(T_1) \notin \Omega_{T_1,R^2 h(T_1)}^{\mathcal{V}}$. As a result, for all $x(t_0) \in \mathcal{X}_0 \times \{0^{n_\psi}\}$, we have $x(t) \in \Omega_{t,R^2h(t)}^{\mathcal{V}}$, for all $t \in [t_0, T]$, and thus $x(T) \in \Omega_{T,R^2}^{\mathcal{V}}$. Finally, it follows from \eqref{eq:constraintG3} that $x_G(T) \in \Omega_{\alpha}^q$. 
\end{proof}

\begin{remark}
	The use of soft IQCs requires some care as they are only defined in the frequency domain for $\mathcal{L}_2$ inputs and yet the analysis must be performed using $\mathcal{L}_{2e}$ signals (to prevent circular arguments). Section~\ref{sec:reach_perturb} is restricted to the use of hard IQCs for which the time-domain IQC holds over finite time horizons.  This removes the technical details associated with soft IQCs. This restricts the analysis to IQCs that can be parameterized so that they are hard. In Section~\ref{sec:soft_IQC}, however, analysis conditions are derived based on soft IQCs. The issues related to soft IQCs are resolved by constructing a finite horizon lower bound valid for $\mathcal{L}_{2e}$ signals (Lemma~\ref{lemma:iqc_bound}). This lower bound is then incorporated in the reachability analysis in Theorem~\ref{thm:soft_IQC}. The proof of Theorem~\ref{thm:soft_IQC} demonstrates that that the reachability analysis uses the lower bound \eqref{eq:IQC_lowerbound} valid for $\mathcal{L}_{2e}$ signals $(v,l,z)$, instead of using \eqref{eq:softIQC}, which requires $(v,l,z)$ to be $\mathcal{L}_2$ signals.
	
	Note that the characterization of a frequency domain IQC as ``soft" vs. ``hard" depends on the factorization of the frequency domain multiplier.  The J-spectral factorization in \cite{Pete_IQC} always yields a ``hard" IQC for any frequency domain multiplier (although this may not be an ideal parameterization for numerical implmentations)
\end{remark}

By applying the generalized S-procedure \cite{Parrilo:00} to \eqref{eq:cond_softIQC}, and using $\alpha$ as the cost function, we obtain the following SOS problem, $\boldsymbol{sosopt_3(F, H, p, g, q, r_0, R, h, p_\delta, \Psi, \mathcal{M})}$:
\begingroup
\allowdisplaybreaks
\begin{subequations}
\begin{align} 
\min_{\alpha, \eta, s, V, M, Y_{22}, \epsilon_1, \epsilon_2} & \ \ \alpha \nonumber \\
\text{s.t.} \ \ \ \ \ 
& V \in \mathbb{R}[(t,x)], M \in \mathcal{M} \ \text{and} \ Y_{22} \in \mathbb{S}^{n_\psi} \ \text{satisfying} \ \eqref{eq:cond_Y22}, \nonumber \\
&s_5 - \epsilon_1 \in \Sigma[x], s_6 - \epsilon_2 \in \Sigma[(x,t)],\epsilon_1 > 0, \epsilon_2 > 0,  \nonumber \\
&s_4 \in \Sigma[x_G], \ s_7 \in \Sigma[(x,t)], \ s_i \in \Sigma[(x,w,l,t,\delta)], \ \forall i \in \{1,2,3\},  \nonumber \\
& -\bigg(\frac{\partial  V}{\partial t} + \frac{\partial  V}{\partial x} F + z^\top Mz- w^\top w \bigg)   + \revise{(p - \eta)}s_1 -s_2g - s_3 p_\delta  \in \Sigma[(x,w,l,t,\delta)],  \\
& -V\vert_{t = t_0, x = [x_G; 0^{n_\psi}]} + s_4 r_0 \in \Sigma[x_G],  \\
& -( q - \alpha)s_5 + \mathcal{V}\vert_{t = T}  - R^2 \in \Sigma[x],  \\
& - \revise{(p-\eta)}s_6 + \mathcal{V} - R^2h  - s_7 g \in \Sigma[(x,t)]. 
\end{align}
\end{subequations}
\endgroup
Compared with $sosopt_2$, the optimization $sosopt_3$ has one more decision matrix $Y_{22}$ and an associated KYP LMI convex constraint, and it can also be solved by using Algorithm~\ref{alg:sol_algo}.}

\section{Examples} \label{sec:examples}
A workstation with a 2.7 [GHz] Intel Core i5 64 bit processors and 8[GB] of RAM was used for performing all computations in the following examples. The SOS optimization problem is formulated and translated into SDP using the Sum-of-Squares module in Yalmip \cite{Yalmip:04} on MATLAB, and solved by the SDP solver Mosek \cite{Mosek:17}. Table \ref{tab:table} shows the degree of polynomials we chose, and the computation time it took for each example.

\begin{table}[H]
	\caption{Computation times for each example \label{tab:table}}
	\centering
	\begin{tabular}{ |M{5.2cm}|M{2cm}|M{2cm}|M{2cm}|M{2cm}|M{1.4cm}|}
		\hline
		Examples / sections & \# of $x$ &Degree of $f$ &Degree of $V$ &Degree of $s$ &Time[sec] \\
		\hline
		Section \ref{ex1}    &2 & 3 &8 & 6 & $6.1 \times 10^1$\\ \hline 
		Section \ref{sec:softhard_compare} &4&3&6&4&$1.1\times 10^2$ \\ \hline
		Section \ref{ex3}: GTM     &4 & 6 &8 &6 &$1.1 \times 10^3$\\ \hline
		Section \ref{ex4}: GTM with $w$     &4&3&8&6&$3.2 \times 10^3$ \\ \hline
		Section \ref{ex5}: GTM with $w$, $\delta$ &4&3&8&6& $5.0 \times 10^3$ \\ \hline
		Section \ref{ex6}: GTM with $w$, $\Delta$  &6&3&6&4&$8.2 \times 10^3$ \\ \hline
		Section \ref{ex7} &7&3&6&6&$3.7 \times 10^3$ \\ \hline
	\end{tabular}
\end{table}
The dynamics $f$ in the following examples are all time-invariant, but since our reachability analysis is addressed in finite-time horizon, we use time-varying storage functions.

\subsection{Van der Pol example} \label{sec:softhard_compare}
Consider the following Van der Pol oscillator dynamics in reverse time with time-invariant uncertain parameter $\delta_{TI} \in [-3, 3]$:
\begin{equation}
\begin{array}{llll}
&\dot{x}_1 = x_2(1 + 0.2 \delta_{TI}), \nonumber \\
&\dot{x}_2  = x_1 + (x_1^2 - 1)x_2. \nonumber 
\end{array}
\end{equation}
In this case $\delta_{TI}$ is treated as a perturbation, where $l = \Delta(v) = \delta_{TI} v$, and $v = 0.2 x_2$. As discussed in Example \ref{ex:IQC_example_para}, the time invariant uncertain parameter $\delta_{TI}$ satisfies both hard and soft IQCs: $\delta_{TI} \in$ HardIQC$(\Psi, M_D)$ and $\delta_{TI} \in$ SoftIQC$(\Psi, M_{DG})$, where the constraints for $M_D$ and $M_{DG}$ are given in Example~\ref{ex:IQC_example_LTI} and \ref{ex:IQC_example_para}, respectively. The robust reachability analysis is performed using both kinds of IQCs. In both cases, we use the same filter $\Psi$, and choose $d$ and $m$ from \eqref{eq:filter_choice} to be $d = 1$, $m = 4$, which correspond to $\Psi$ described by the following dynamics:
\begin{align}
&A_\psi = \left[\begin{smallmatrix}-4 & 0 \\ 0 & -4 \end{smallmatrix}\right], ~~ B_{\psi 1} = \left[\begin{smallmatrix}1 \\ 0 \end{smallmatrix}\right], ~~~~~~~~ B_{\psi 2} = \left[\begin{smallmatrix}0 \\ 1 \end{smallmatrix}\right], \nonumber \\
&C_\psi = \left[\begin{smallmatrix}0, 1, 0,0 \\ 0, 0, 0, 1\end{smallmatrix}\right]^\top, D_{\psi 1} = \left[\begin{smallmatrix}1, 0, 0,0 \end{smallmatrix}\right]^\top, D_{\psi 2} = \left[\begin{smallmatrix}0, 0, 1,0 \end{smallmatrix}\right]^\top. \nonumber 
\end{align}
Therefore the filter $\Psi$ introduces two filter states $x_\psi \in \R^2$ to the extended system. Take the time horizon as $[t_0, T] = [0, 1.5]$ and the initial set as $\mathcal{X}_0 = \{(x_1,x_2)~|~ x_1^2 + x_2^2 \leq 1\}$. Choose polynomials $q = p = 0.3150x_1^2 - 0.0976 x_1 x_2 
+ 0.0816 x_2^2 - 0.0023 x_1 
+ 0.0002 x_2$. The local region $\mathcal{X}_l$ is picked as $\Omega_{4}^{p}$. The optimal $\alpha$ computed using soft and hard IQCs are 1.21 and 1.60, respectively, which states the fact that the soft IQC achieves a less conservative outer bound and captures the nature of the uncertainty. In Fig~\ref{fig:vdp}, the simulation points $x(T)$ of the Van der Pol dynamics with the initial set $\mathcal{X}_0$, and with values of $\delta_{TI}$ randomly drawn from $[-3,3]$ are shown with green dots. We can see from Fig~\ref{fig:vdp} that the outer bound obtained using the soft IQC (shown with the black solid curve) is enclosed by the one computed using the hard IQC (shown with the purple dash-dotted curve). It also indicates that the outer bound obtained using the soft IQC is less conservative.
\begin{figure}[h]
\centering
\includegraphics[width=0.4\textwidth]{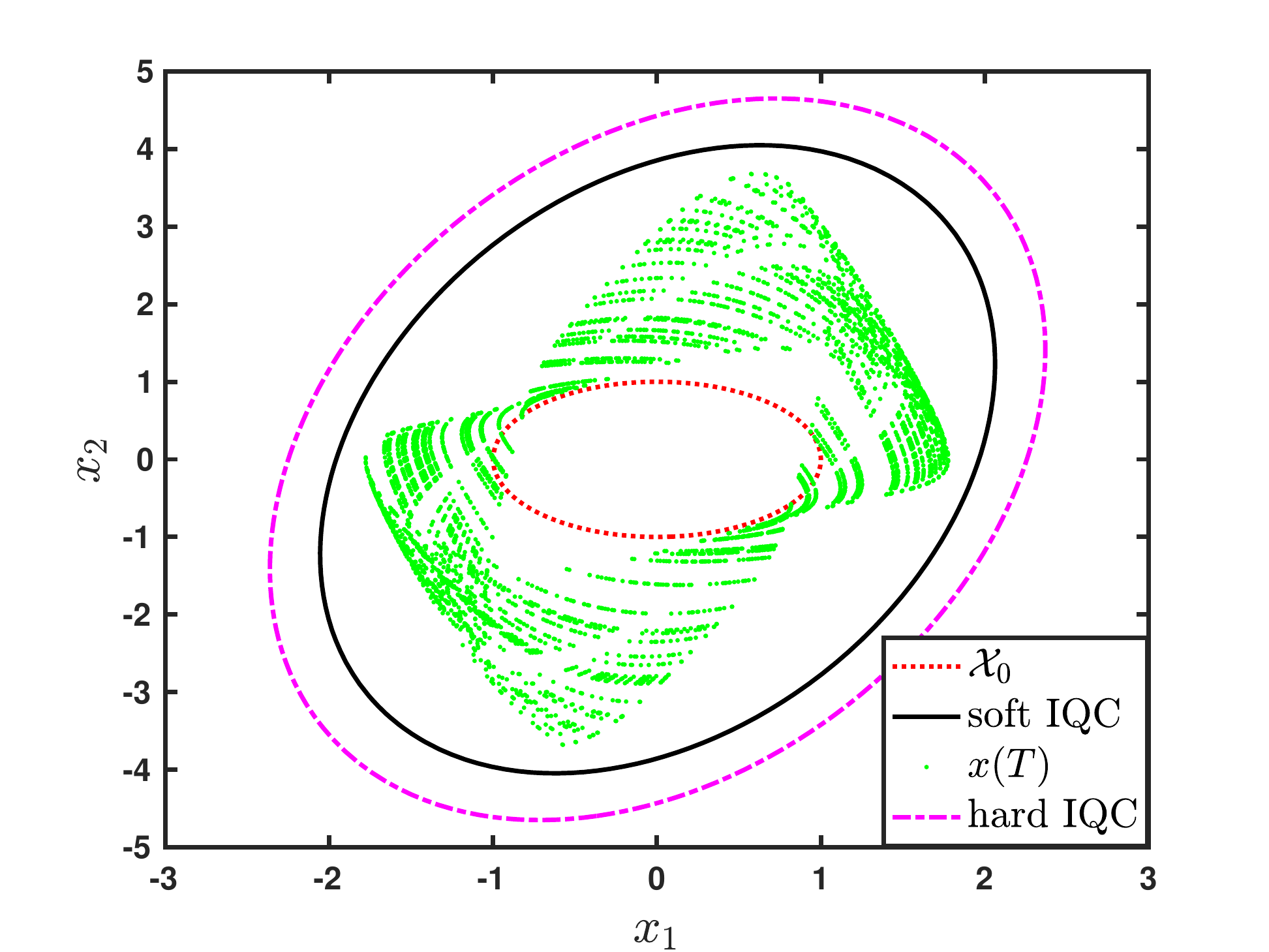}
\caption{Outer bounds using soft/hard IQCs and simulation points $x(T)$ at $T = 1.5$ under uncertain parameter, with the initial condition set $\mathcal{X}_0$}
\label{fig:vdp}    
\end{figure}

\subsection{NASA's Generic Transport Model (GTM) around straight and level flight condition} 
The GTM is a remote-controlled $5.5 \%$ scale commercial aircraft  \cite{Murch:07}. Its longitudinal model  \cite{Stevens:92} is
\begin{align}
\dot{x}_1 =& \frac{1}{m}(-D -mg\sin(x_4 - x_2) + T_x \cos(x_2) + T_z \sin(x_2)), \nonumber \\
\dot{x}_2 =& \frac{1}{m x_1}(-L +mg \cos(x_4 - x_2) - T_x \sin(x_2) + T_z \cos(x_2) + x_3), \nonumber \\
\dot{x}_3 =& \frac{M + T_m}{I_{yy}}, \label{eq:q} \\
\dot{x}_4 =& x_3, \nonumber 
\end{align}
\noindent where states $x_1$ to $x_4$ represent air speed (m/s), angle of attack (rad), pitch rate (rad/s) and pitch angle (rad) respectively. The control inputs are elevator deflection $u_{elev}$ (rad) and engine throttle $u_{th}$ (percent). The drag force $D$ (N), lift force $L$ (N), and aerodynamic pitching moment $M$ (Nm) are given by $D = \bar{q}SC_D(x_2, u_{elev}, \hat{q})$, $L=\bar{q}SC_L(x_2,u_{elev},\hat{q})$, and $M = \bar{q}S\bar{c}C_m(x_2, u_{elev}, \hat{q})$,where $\bar{q}:= \frac{1}{2}\rho x_1^2$ is the dynamic pressure (N/m$^2$), and $\hat{q}:=(\bar{c}/2x_1)x_3$ is the normalized pitch rate (unitless). $C_D, C_L,$ and $C_m$ are unitless
aerodynamic coefficients computed from look-up tables provided by NASA.

A degree-6 polynomial model, provided in \cite{Chakraborty:11}, is obtained after replacing all nonpolynomial terms with their polynomial approximations. The polynomial model takes the form $\dot{x} = f_6(x, u)$, where $x:= [x_1, x_2, x_3, x_4]^\top$ and $u = [u_{elev}, u_{th}]^\top$. The following straight and level flight condition is computed for this model: $x_{1,t} = 45$ m/s, $x_{2,t} = 0.04924$ rad, $x_{3,t} = 0$ rad/s, $x_{4,t} = 0.04924$ rad, with $u_{elev,t} = 0.04892$ rad, and $u_{th,t} = 14.33 \%$. The subscript $t$ denotes a trim value. A polynomial closed-loop model, denoted as $\dot{x} = f_6(x)$, is obtained by holding $u_{th}$ at its trim value, applying a proportional pitch rate feedback $u_{elev} = K_q x_3 + u_{elev,t} = 0.0698x_3 + u_{elev,t}$.

\subsubsection{Analysis for the GTM} \label{ex3}
Reachability analysis is carried out on $\dot{x} = f_6(x)$ around its trim point. The set of initial conditions $\mathcal{X}_0 = \{x\in \mathbb{R}^4| (x - x_t)^\top C^{-1} (x - x_t) -1 \le 0\}$ is a 4-dimensional ellipsoid inside the region of attraction, where $C = diag(20^2, (20\pi/180)^2,$ $(50\pi/180)^2, (20\pi/180)^2)$, $x_t$ is the trim point. Take the local region $\mathcal{X}_l = \{x \in \mathbb{R}^4 | (x - x_t)^\top C_l^{-1} (x - x_t)-1 \leq 0  \}$, where $C_l = diag(30^2, (30\pi/180)^2, (75\pi/180)^2, (30\pi/180)^2)$. $\Omega_1^q$ is chosen as the minimum volume ellipsoid containing all the simulation points at terminal time.

To improve the numerical conditioning, we define the scaled states $x_{scl} = N_{scl}^{-1}x$, where we set $N_{scl} := diag(20,$ $20\pi/180,$ $50\pi/180,$ $20\pi/180)$, since $20$ m/s, $20 \pi /180 $ rad, $50 \pi/180$ rad/s, $20 \pi /180$ rad are farthest distances observed in simulation that each state can be away from their trim point value given the initial condition set $\mathcal{X}_0$. Then we have the dynamics for the scaled states: $\dot{x}_{scl} = N_{scl}^{-1}f_6(N_{scl}x_{scl})$, and this scaled dynamics is the one we will use in the SOS optimization problem. Before scaling, the coefficients of $f_6(x)$ vary from $1.6 \times 10^{-5}$ to $4.5 \times 10^1$; after scaling, they vary from $4.5 \times 10^{-3}$ to $1.8 \times 10^1$. Before plugging the polynomial functions $r_0, q, p$ into the SOS optimization problem, the parameters were scaled accordingly.

Figure \ref{fig:GTMx2x3} and Figure \ref{fig:GTMx1x4} show the outer bound of reachable set in $x_2 - x_3$ space and $x_1 - x_4$ space respectively, at different simulation times. We can observe that $\Omega_{T,0}^V$ (black curve) contains all the simulation points $x(T)$ (green points) at each terminal time $T$.
\begin{figure}[H]
	\centering
	\includegraphics[width=0.7\textwidth]{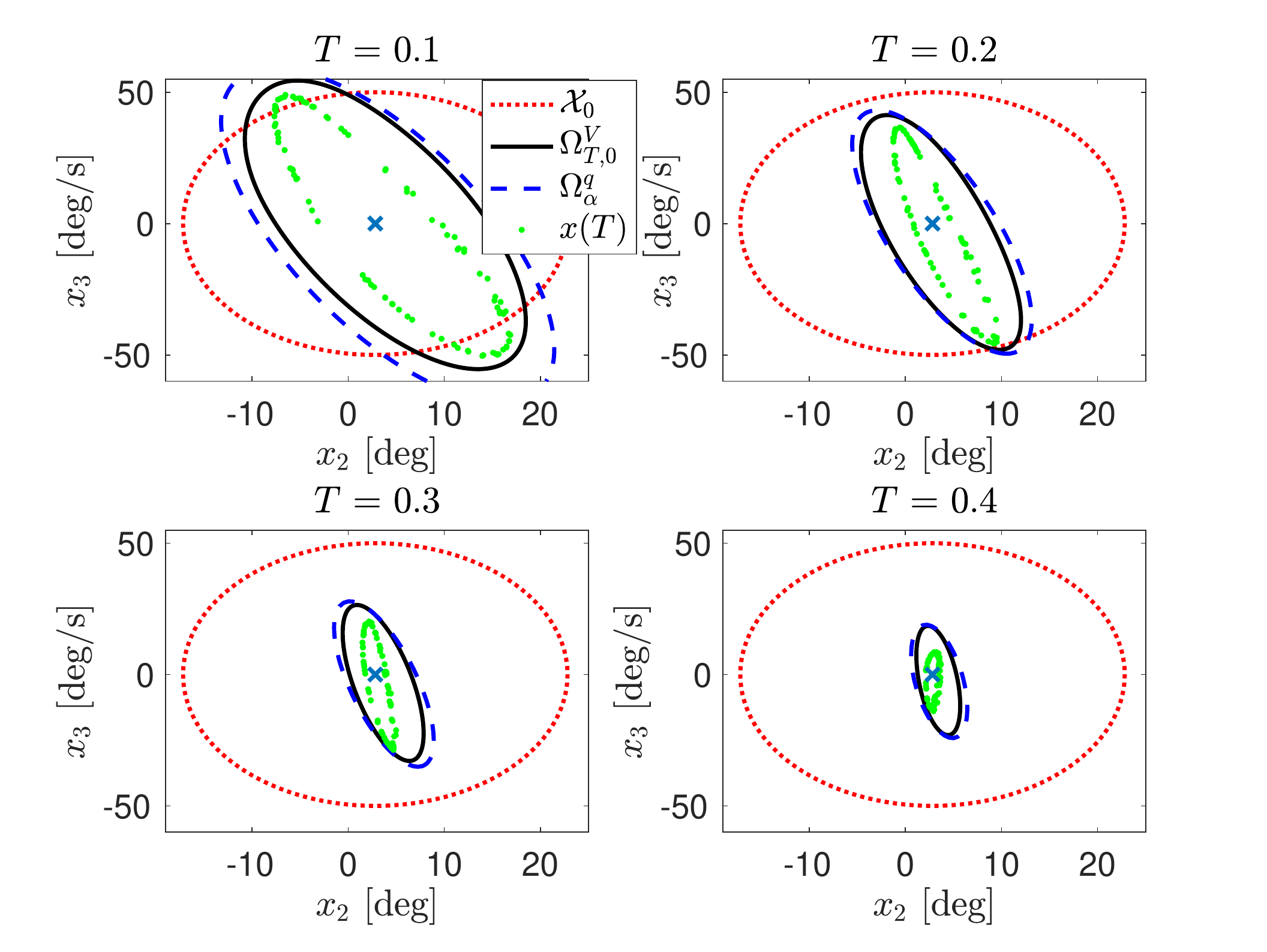}
	\caption{Outer bounds for GTM model in $x_2 - x_3$ plane.}
	\label{fig:GTMx2x3}    
\end{figure}
\begin{figure}[H]
	\centering
	\includegraphics[width=0.7\textwidth]{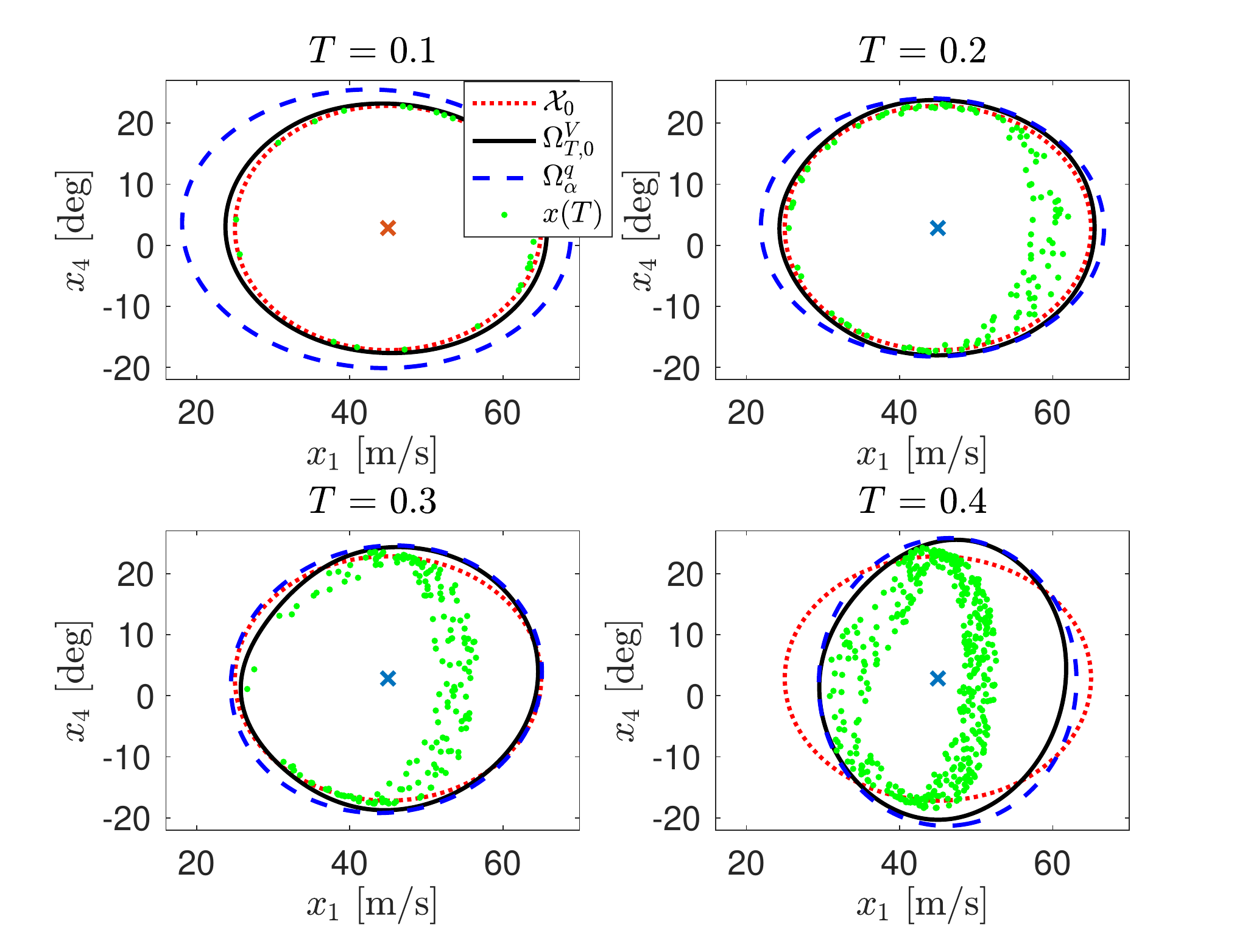}
	\caption{Outer bounds for GTM model in $x_1 - x_4$ plane.}
	\label{fig:GTMx1x4}    
\end{figure}


\subsubsection{GTM with $\mathcal{L}_2$ disturbance}  \label{ex4}
To save computation time, reachability analysis is conducted on a 4-state degree-3 model obtained from the 4-state degree-6 model, with the same initial condition set $\mathcal{X}_0$ as that from the previous section. But an input disturbance $w$ at the elevator channel is taken into consideration this time. The control input becomes $u_{elev} = K_q x_3 + u_{elev,t} + w = 0.0698x_3 + u_{elev,t} + w$. Figure \ref{fig:disturbance} shows outer bounds at time $T = 0.4$ s with disturbances of different $\mathcal{L}_2$ bounds.

\begin{figure}[h]
	\centering
	\includegraphics[width=0.7\textwidth]{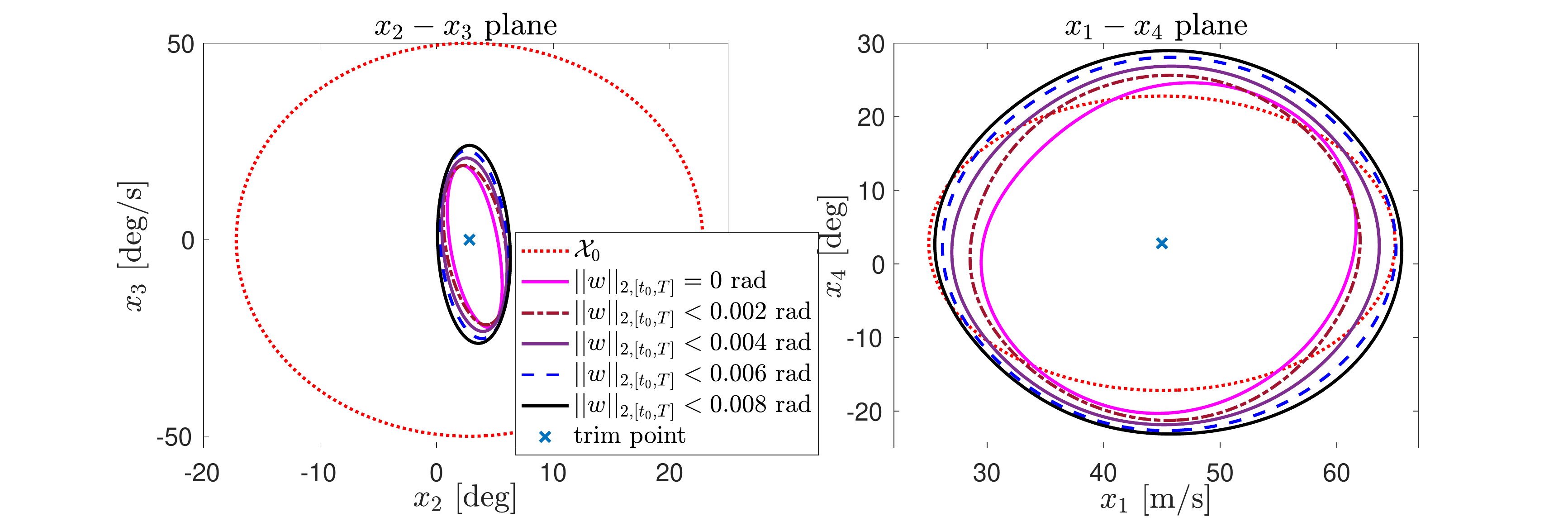}
	\caption{Outer bounds for GTM model at $T = 0.4$ sec with $\mathcal{L}_2$ disturbances. }
	\label{fig:disturbance}    
\end{figure}

\subsubsection{GTM with $\mathcal{L}_2$ disturbances and time varying uncertain parameters} \label{ex5}
In addition to an input disturbance $w$ at the elevator channel, satisfying $\norm{w}_{2,[t_0,T]} < 0.004$ rad, assume that the inertia $I_{yy}$ in (\ref{eq:q}) is also uncertain: $I_{yy} = \gamma(t) \bar{I}_{yy}$, where $\gamma(t)$ is a time varying uncertain parameter  and $\bar{I}_{yy}$ is the nominal value of inertia. Define $\delta := 1/\gamma$, assume $\gamma(t) \in [\frac{10}{11}, \frac{10}{9}]$, then $\delta(t) \in [0.9, 1.1] =: \mathcal{D}$. Equation (\ref{eq:q}) becomes
\begin{align}
\dot{x}_3 =&  \frac{M + T_m}{I_{yy}} = \frac{M + T_m}{\gamma \bar{I}_{yy}}  = \delta \frac{M + T_m}{\bar{I}_{yy}}. \nonumber 
\end{align}
The result is shown in Figure \ref{fig:parameter}, where the outer bounds with and without uncertain parameter are shown with blue and magenta curves, respectively.
\begin{figure}[h]
	\centering
	\includegraphics[width=0.7\textwidth]{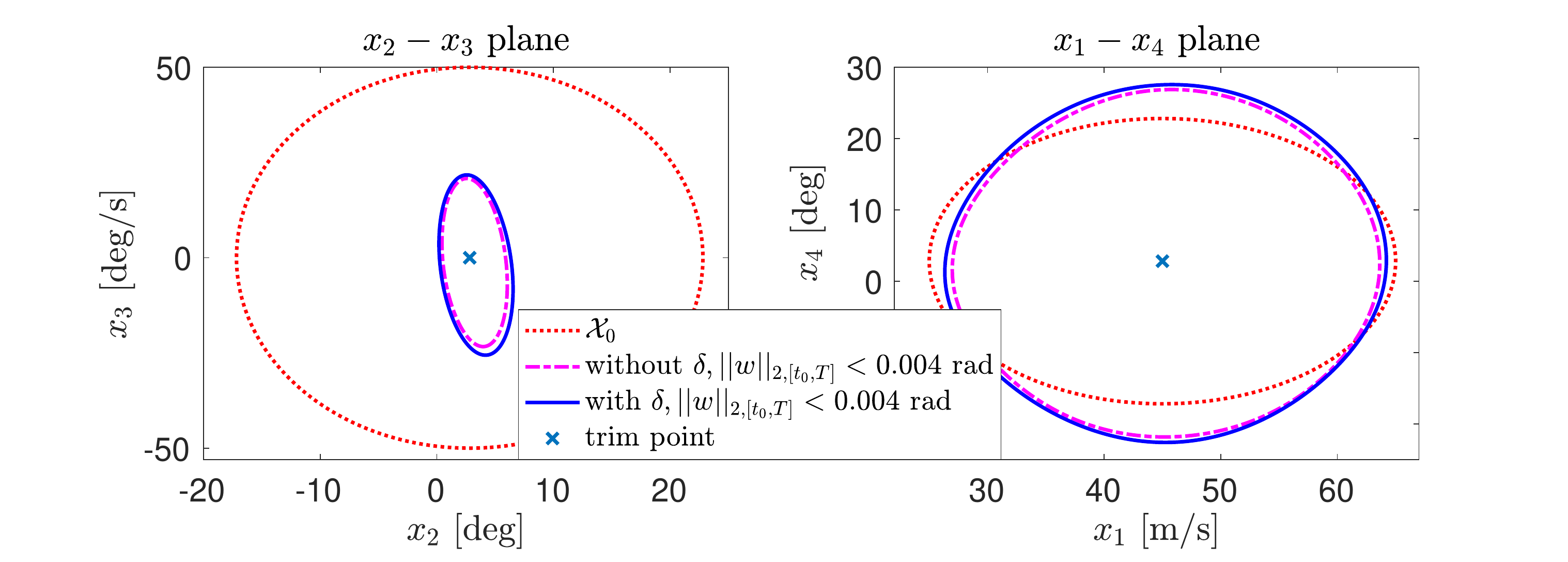}
	\caption{Over bounds for GTM model at $T = 0.4$ sec with $\mathcal{L}_2$ disturbance and parameter.}
	\label{fig:parameter}    
\end{figure}

\subsubsection{GTM with $\mathcal{L}_2$ disturbance and unmodeled dynamics $\Delta$} \label{ex6}
Assume the control input at elevator actuator of the GTM system is corrupted by an $\mathcal{L}_2$ disturbance $w$, satisfying $\norm{w}_{2,[t_0,T]} < 0.004$ rad, and an LTI uncertainty $\Delta$ with $\norm{\Delta}_{\infty} \leq \sigma$, where $\sigma > 0$. Figure \ref{fig:GTMFu} shows a block diagram for the uncertain GTM system. The input to the perturbation is $v = K_q x_3 + u_{elev,t}+ w$, and the signal that actually goes into the elevator channel is $u_{elev} = v + l$, where $l = \Delta(v)$. As discussed in Example~\ref{ex:IQC_example_LTI}, this LTI uncertainty $\Delta$ satisfies hard IQCs defined by $(\Psi, M_D)$ from Example \ref{ex:IQC_example_LTI}. In this example, we choose $d$ and $m$ from \eqref{eq:filter_choice} to be $d = 1$, $m = 1$, and they correspond to $\Psi$ of the following dynamics:
\begin{align}
&A_\psi = \left[\begin{smallmatrix}-1 & 0 \\ 0 & -1 \end{smallmatrix}\right], ~~ B_{\psi 1} = \left[\begin{smallmatrix}1 \\ 0 \end{smallmatrix}\right], ~~~~~~~~ B_{\psi 2} = \left[\begin{smallmatrix}0 \\ 1 \end{smallmatrix}\right], \nonumber \\
&C_\psi = \left[\begin{smallmatrix}0, 1, 0,0 \\ 0, 0, 0, 1\end{smallmatrix}\right]^\top, D_{\psi 1} = \left[\begin{smallmatrix}1, 0, 0,0 \end{smallmatrix}\right]^\top, D_{\psi 2} = \left[\begin{smallmatrix}0, 0, 1,0 \end{smallmatrix}\right]^\top. \nonumber 
\end{align}
Again, the filter $\Psi$ introduces two filter states $x_\psi \in \R^2$ to the extended system. We solve for the outer bounds with two values of $\sigma$ using $sosopt_2$ \revise{with the constraint set $\mathcal{M}_1$ defined in \eqref{eq:MD}}. The results are shown in Fig \ref{fig:perturbation}, where the outer bound with $\sigma = 0.1$ is shown with the magenta curve, and the one with $\sigma = 0.4$ is shown with the blue curve.
\begin{figure}[h]
	\centering
	\includegraphics[width=0.5\textwidth]{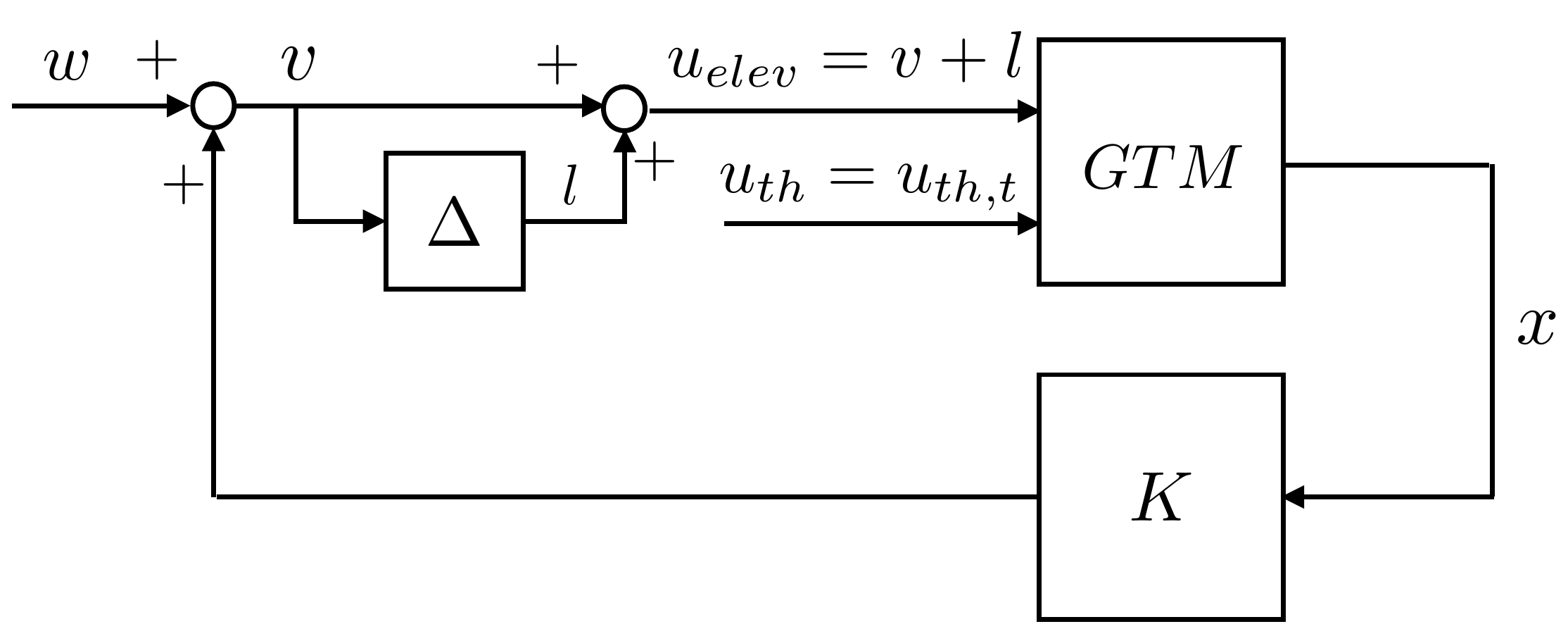}
	\caption{Uncertain nonlinear model for GTM.}
	\label{fig:GTMFu}    
\end{figure}
\begin{figure}[h]
	\centering
	\includegraphics[width=0.7\textwidth]{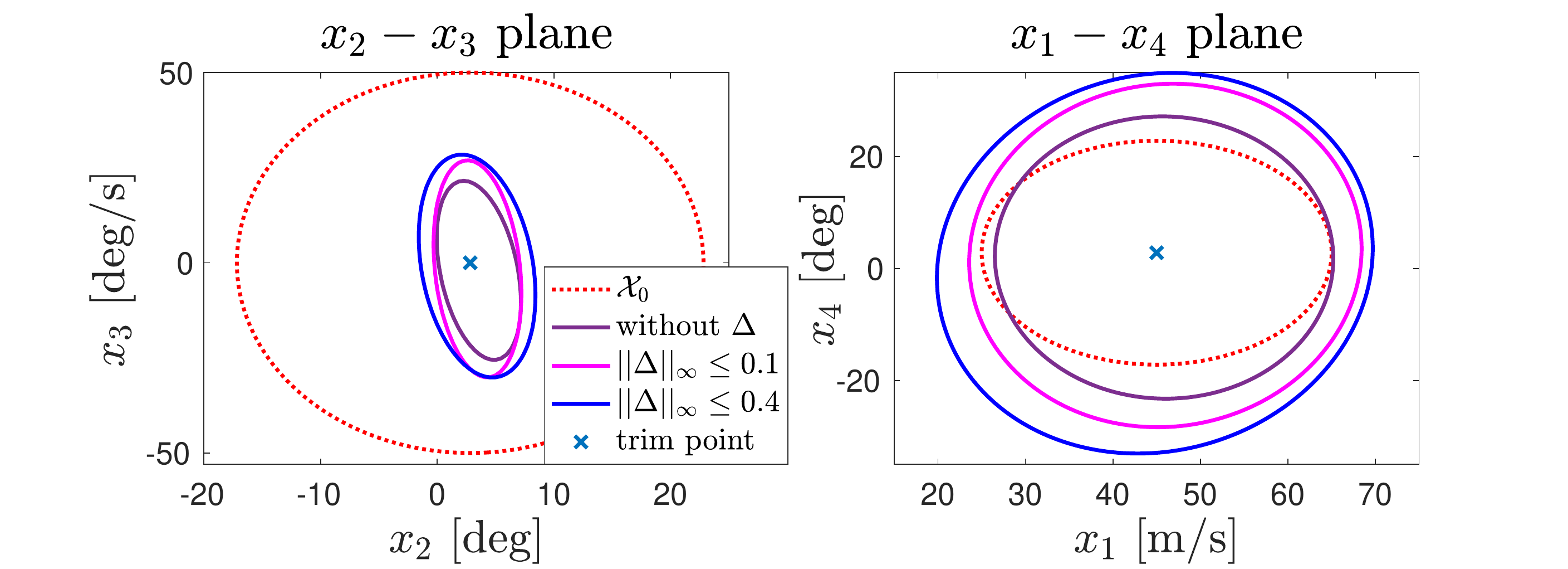}
	\caption{Outer bounds for GTM model at $T = 0.4$ sec with $\mathcal{L}_2$ disturbance and perturbation $\Delta$. }
	\label{fig:perturbation}    
\end{figure}

\subsection{F-18 around falling-leaf mode flight condition} \label{ex7}
In this example, we analyze a 7-state, cubic degree F-18 closed-loop polynomial model $\dot{x}= f_{3}(x)$ from \cite{Seiler:11}, where the states $x_1, ..., x_7$ represent sideslip angle (rad), angle-of-attack (rad), roll rate (rad/s), pitch rate (rad/s), yaw rate (rad/s), bank angle (rad), controller state (rad) respectively. The trim point of the states is $x_t = $ [0 \text{degree}, 20.17 \text{degree}, -1.083 \text{degree/s}, 1.855 \text{degree/s}, 2.634 \text{degree/s}, 35 \text{degree}, 0 \text{degree}]. Consider the flight condition for a coordinated turn ($x_{1,t} = 0^{\circ}$) at a $35^{\circ}$ bank angle and at the air speed of 350 ft/s. Around this condition the aircraft is more likely to experience the falling-leaf motion.  The analysis is performed around this flight condition.

The given initial set $\mathcal{X}_0 = \{x \in \mathbb{R}^7 | (x - x_t)^\top C^{-1} (x - x_t) -1\le 0\}$ is a 7-dimensional ellipsoid inside the region of attraction, where $C = diag((10\pi/180)^2, $$(25\pi/180)^2,$ $(35\pi/180)^2,$ $(30\pi/180)^2,$ $(15\pi/180)^2,$ $(25\pi/180)^2,$ $(20$$\pi/$$180)^2)$. Again, in order to improve the numerical conditioning, we scale the states $x_{scl} = N^{-1}x$, where $N = diag(10\pi/180,$ $25\pi/180,$ $35\pi/180,$ $ 5\pi/180,$ $15\pi/180,$ $25\pi/180,$ $20$$\pi/180)$. Take $\mathcal{X}_l$ with radii twice as long as those of $\mathcal{X}_0$. Take $\Omega_1^q$ as the minimum volume ellipsoid containing all the simulation points $x(T)$.

Figure \ref{fig:F18x1x2} and Figure \ref{fig:F18x3x5} show the outer bound of reachable set in $x_1 - x_2$ space and $x_3- x_5$ space respectively, at different simulation times. The red dotted curve is a slice of initial set $\mathcal{X}_0$. We can see that $\Omega_{T,0}^V$ (shown with the black curve) tightly contains $x(T)$ (shown with green points). We verified the reliability of the solutions of this example from SOS programming with a post-processing step as advocated in \cite{Lofberg:09}. 
\begin{figure}[H]
	\centering
	\includegraphics[width=0.7\textwidth]{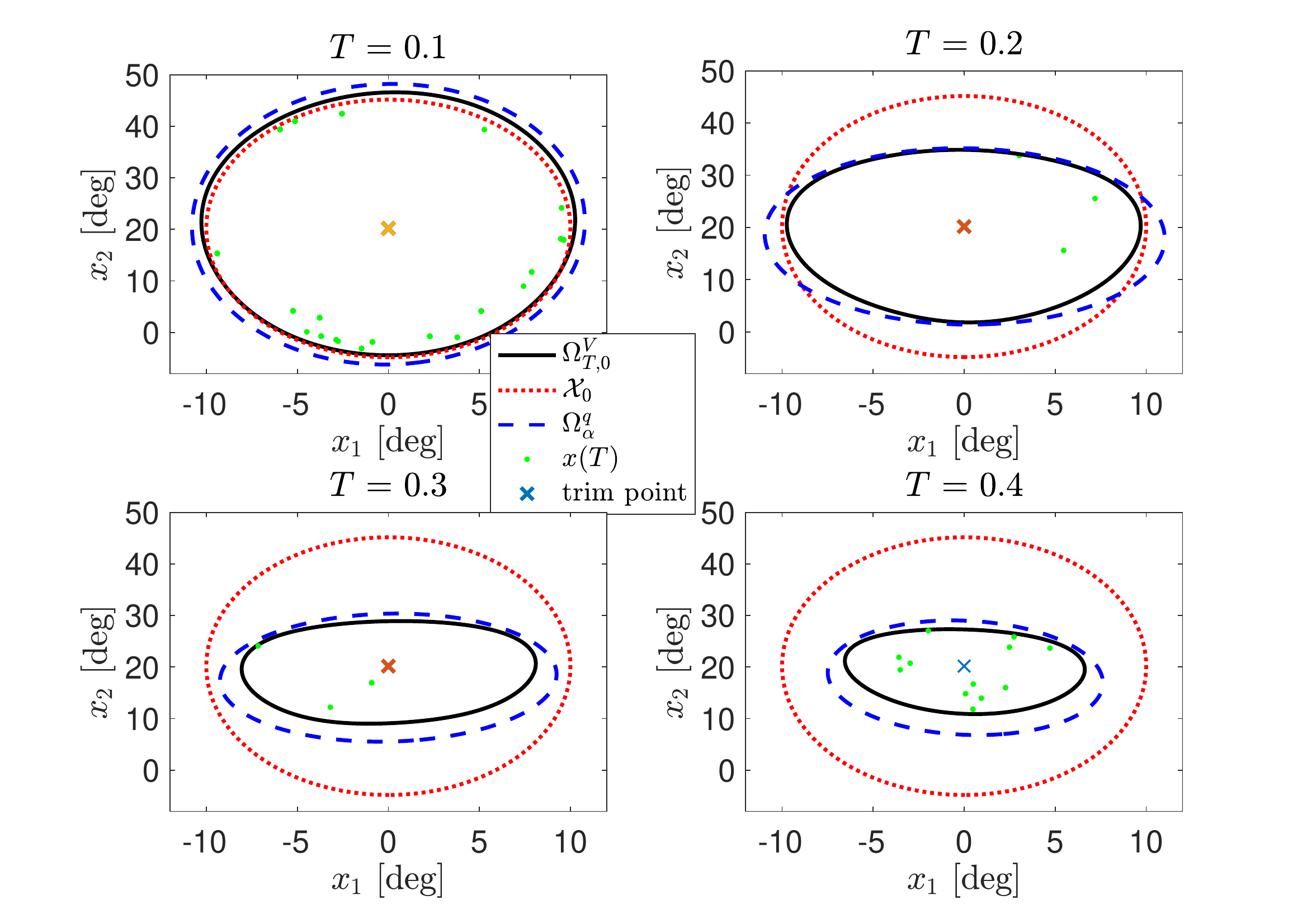}
	\caption{Outer bounds for F-18 model in $x_1 - x_2$ plane.}
	\label{fig:F18x1x2}    
\end{figure}
\begin{figure}[H]
	\centering
	\includegraphics[width=0.7\textwidth]{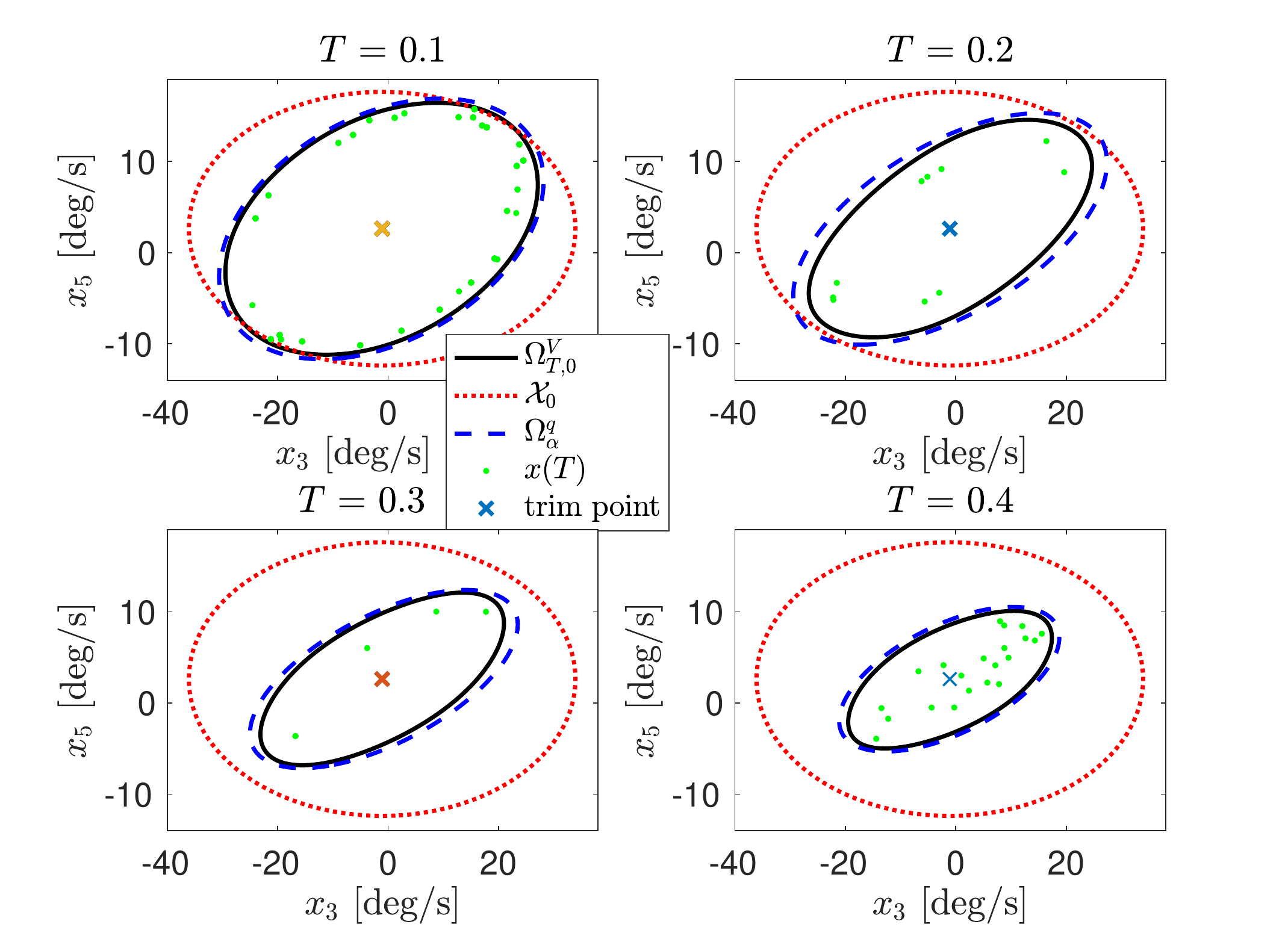}
	\caption{Outer bounds for F-18 model in $x_3 - x_5$ plane.}
	\label{fig:F18x3x5}    
\end{figure}

\subsubsection{Comparison to the $V, s$ iterations method} \label{sec:compare_Vs}
	The outer bound of the reachable set at $T = 0.4$ sec is also computed using the $V, s$ iterations method from \cite{Majumdar:17} with the same shape function $q$ as the one we used before. The outer-approximations obtained using the $V,s$ iterations and quasi-convex methods are shown with the brown curves and black curves in Figure \ref{fig:compare}, respectively. We can see that the brown curves enclose the black curves in both plots, and thus the outer bound from the quasi-convex method is less conservative.
\begin{figure}[h]
	\centering
	\includegraphics[width=0.7\textwidth]{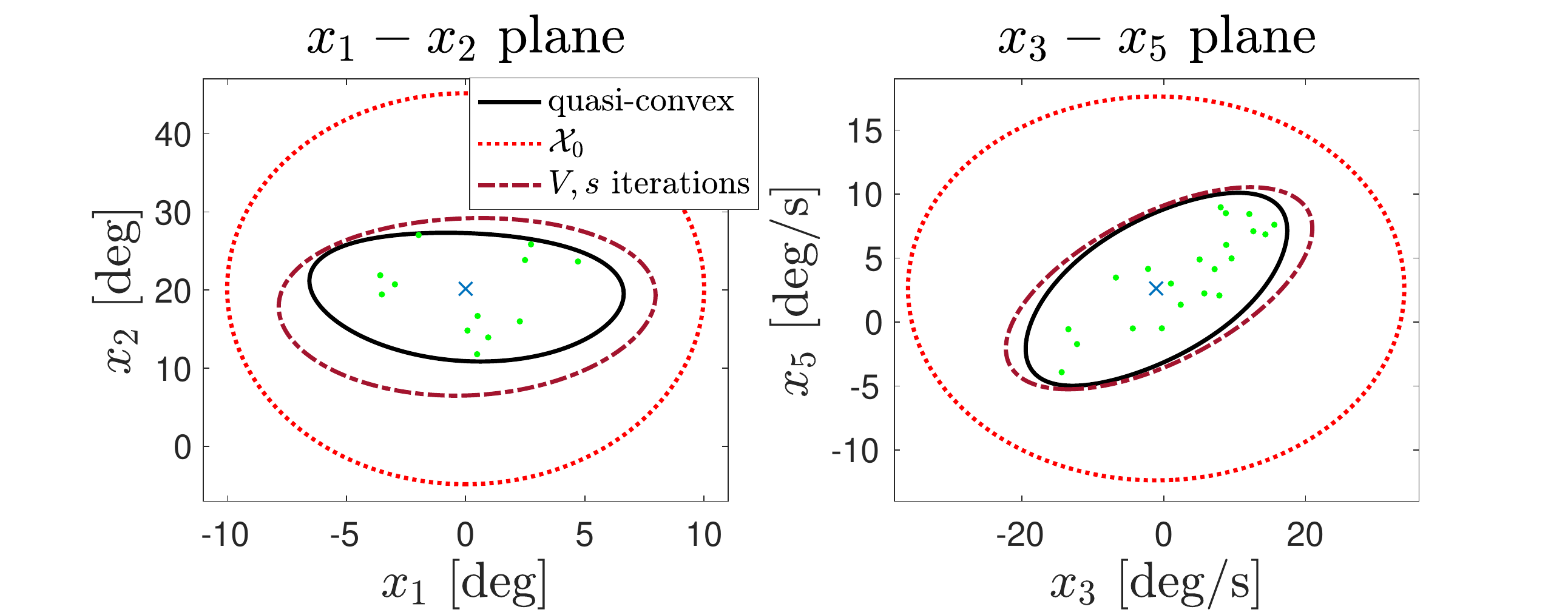}
	\caption{Comparison of outer bounds at $T = 0.4$ sec for F-18 model. }
	\label{fig:compare}    
\end{figure}

The computation details are shown in Table \ref{tab:table}, including the obtained $\alpha$ and computation time. We can see that compared with the $V, s$ iterations, within the similar amount of computation time, using the same shape function, the quasi-convex method from this paper is able to achieve smaller $\alpha$, i.e. less conservative outer bound. Also, from the value of $\alpha^\star$ reported in Table \ref{tab:table}, we can see that the outer bound obtained using our method is contained by $\Omega_{1.36}^q$, whose radii are 1.166 times those of  $\Omega_1^q$, the minimum-volume ellipsoid that contains all the simulation points. This indicates the tightness of the outer bound. 
\begin{table}[h]
	\centering
	\caption{Computation results and details for the two methods \label{tab:table}}
	\begin{tabular}{ |M{2.5cm}|M{1cm}|M{2cm}|M{2cm}|M{2cm}|}
		\hline
		Methods &$\alpha^*$ &Degree of $V$ &Degree of $s$ &Time[sec] \\
		\hline
		$V, s$   iterations &1.70 & 4 &4  & $5.2 \times 10^3$\\ \hline 
		quasi-convex     &1.36 & 6 &6  &$3.7 \times 10^3$\\ \hline
	\end{tabular}
\end{table}

\section{Conclusions} \label{sec:conclu}
%
We proposed a method for computing outer bounds of reachable sets using time varying storage functions that satisfy ``local" dissipation
inequalities. The method is developed for nonlinear systems with polynomial vector fields and simultaneously accounts for $\mathcal{L}_2$ disturbances,
parametric uncertainties, and perturbations described by time-domain integral
quadratic constraints (IQCs). A key aspect is that IQCs can be used to account for unmodeled dynamics. The computational algorithms rely on SOS programming
and the generalized S-procedure. This leads to quasi-convex optimizations for computing the tightest outer bound of the reachable set. It is thus possible to compute the global optima for this optimization and no initialization is required for the storage function. We applied the proposed method to several examples including several using nonlinear aircraft models.
\section*{Acknowledgements}  
This work was funded in part by the grants ONR grant N00014-18-1-2209, AFOSR FA9550-18-1-0253, and NSF ECCS-1906164. A. Packard acknowledges generous support from the FANUC Corporation through the FANUC Chair in Mechanical Systems.

\bibliography{reference.bib}

\begin{thebibliography}{10}

\bibitem{Balakrishnan:02}
V.~{Balakrishnan}.
\newblock Lyapunov functionals in complex /spl mu/ analysis.
\newblock {\em IEEE Transactions on Automatic Control}, 47(9):1466--1479, Sep.
  2002.

\bibitem{Borrelli:11}
Francesco Borrelli, Alberto Bemporad, and Manfred Morari.
\newblock {\em Predictive control for linear and hybrid systems}.
\newblock Cambridge University Press, 2011.

\bibitem{Balas:11}
Abhijit Chakraborty, Pete Seiler, and Gary Balas.
\newblock Local performance analysis of uncertain polynomial systems with
  applications to actuator saturation.
\newblock In {\em IEEE Conference on Decision and Control}, pages 8176--8181.
  2011.

\bibitem{Seiler:11}
Abhijit Chakraborty, Pete Seiler, and Gary Balas.
\newblock Susceptibility of {F/A}-18 flight controllers to the falling-leaf
  mode: Nonlinear analysis.
\newblock {\em Journal of Guidance, Control, and Dynamics}, 34:73--85, 2011.

\bibitem{Chakraborty:11}
Abhijit Chakraborty, Peter Seiler, and Gary~J. Balas.
\newblock Nonlinear region of attraction analysis for flight control
  verification and validation.
\newblock {\em Control Engineering Practice}, 19(4):335--345, 2011.

\bibitem{Fetzer:2018}
M.~{Fetzer}, C.~W. {Scherer}, and J.~{Veenman}.
\newblock Invariance with dynamic multipliers.
\newblock {\em IEEE Transactions on Automatic Control}, 63(7):1929--1942, July
  2018.

\bibitem{henrion2013convex}
Didier Henrion and Milan Korda.
\newblock Convex computation of the region of attraction of polynomial control
  systems.
\newblock {\em IEEE Transactions on Automatic Control}, 59(2):297--312, 2013.

\bibitem{IANNELLI:2019}
Andrea Iannelli, Peter Seiler, and Andrés Marcos.
\newblock Region of attraction analysis with integral quadratic constraints.
\newblock {\em Automatica}, 109:108543, 2019.

\bibitem{Packard:05}
Zachary Jarvis-Wloszek, Ryan Feeley, Weehong Tan, Kunpeng Sun, and Andrew
  Packard.
\newblock Control applications of sum of squares programming.
\newblock In {\em Positive Polynomials in Control}, volume 312, pages 3--22.
  Springer, Berlin, Heidelberg, 2005.

\bibitem{Jaulin:01}
Luc Jaulin, Michel Kieffer, Olivier Didrit, and Eric Walter.
\newblock {\em Applied interval analysis: with examples in parameter and state
  estimation, robust control and robotics}, volume~1.
\newblock Springer Science \& Business Media, 2001.

\bibitem{MPeet:2019}
M.~{Jones} and M.~M. {Peet}.
\newblock Using sos for optimal semialgebraic representation of sets: Finding
  minimal representations of limit cycles, chaotic attractors and unions.
\newblock In {\em 2019 American Control Conference (ACC)}, pages 2084--2091,
  July 2019.

\bibitem{Varaiya:07}
Alex Kurzhanskiy and Pravin Varaiya.
\newblock Ellipsoidal techniques for reachability analysis of discrete-time
  linear systems.
\newblock {\em IEEE TAC}, 52:26--38, 2007.

\bibitem{Yalmip:04}
Johan L\"{o}fberg.
\newblock {YALMIP} : a toolbox for modeling and optimization in matlab.
\newblock In {\em IEEE International Conference on Robotics and Automation},
  pages 284--289. Taipei, 2004.

\bibitem{Lofberg:09}
Johan L\"{o}fberg.
\newblock Pre-and post-processing sum-of-squares programs in practice.
\newblock {\em IEEE Transactions on Automatic Control}, 54:1007--1011, 2009.

\bibitem{Alessandro:05}
Alessandro Magnani, Sanjay Lall, and Stephen Boyd.
\newblock Tractable fitting with convex polynomials via sum-of-squares.
\newblock In {\em 44th IEEE Conference on Decision and Control}, pages
  1672--1677. Seville, Spain, 2005.

\bibitem{Majumdar:17}
Anirudha Majumdar and Russ Tedrake.
\newblock Funnel libraries for real-time robust feedback motion planning.
\newblock {\em The international Journal of Robotics Research}, 36:947--982,
  2017.

\bibitem{Megretski:97}
Alexandre Megretski and Anders Rantzer.
\newblock System analysis via integral quadratic constraints.
\newblock {\em IEEE Transactions on Automatic Control}, 42:819--830, 1997.

\bibitem{Mitchell:00}
Ian Mitchell and Claire Tomlin.
\newblock Level set methods for computation in hybrid systems.
\newblock In {\em In Hybrid Systems: Computation and Control}, pages 310--323.
  2000.

\bibitem{Mosek:17}
{\relax MOSEK ApS}.
\newblock The {MOSEK} optimization toolbox for {MATLAB} manual. version 8.1.
\newblock 2017.

\bibitem{Murch:07}
Austin Murch and John Foster.
\newblock Recent {NASA} research on aerodynamic modeling of post-stall and spin
  dynamics of large transport airplanes.
\newblock In {\em 45th AIAA Aerospace Sciences Meeting and Exhibit}. Reno,
  Nevada, 2007.

\bibitem{Antonis:02}
A.~{Papachristodoulou} and S.~{Prajna}.
\newblock On the construction of lyapunov functions using the sum of squares
  decomposition.
\newblock In {\em IEEE Conference on Decision and Control}, pages 3482--3487,
  2002.

\bibitem{Parrilo:00}
Pablo Parrilo.
\newblock Structured semidefinite programs and semialgebraic geometry methods
  in robustness and optimization.
\newblock PhD thesis, California Institute of Technology, 2000.

\bibitem{Prajna:04}
Stephen Prajna and Ali Jadbabaie.
\newblock Safety verification of hybrid systems using barrier certificates.
\newblock In {\em Proc. of Hybrid Systems: Computation and Control}, pages
  477--492. Berlin, Heidelberg, 2004.

\bibitem{RANTZER19967}
Anders Rantzer.
\newblock On the {K}alman-{Y}akubovich-{P}opov lemma.
\newblock {\em Systems \& Control Letters}, 28(1):7 -- 10, 1996.

\bibitem{Pete_IQC}
P.~{Seiler}.
\newblock Stability analysis with dissipation inequalities and integral
  quadratic constraints.
\newblock {\em IEEE Transactions on Automatic Control}, 60(6):1704--1709, 2015.

\bibitem{Chris:19}
Peter Seiler, Robert Moore, Chris Meissen, Murat Arcak, and Andrew Packard.
\newblock Finite horizon robustness analysis of ltv systems using integral
  quadratic constraints.
\newblock {\em Automatica}, 100:135--143, 2019.

\bibitem{Stevens:92}
Brian Stevens and Frank Lewis.
\newblock {\em Aircraft Control and Simulation}.
\newblock John Wiley \& Sons, Hoboken, NJ, 1992.

\bibitem{Summers:13}
Erin Summers, Abhijit Chakraborty, Weehong Tan, Ufuk Topcu, Pete Seiler, Gary
  Balas, and Andrew Packard.
\newblock Quantitative local {L}2-gain and reachability analysis for nonlinear
  systems.
\newblock {\em International Journal of Robust and Nonlinear Control},
  23:1115--1135, 2013.

\bibitem{Weehong:06}
Weehong Tan, Andrew Packard, and Timothy Wheeler.
\newblock Local gain analysis of nonlinear systems.
\newblock In {\em Proc. Amer. Control Conf}, pages 92--96. Minneapolis, MN,
  2006.

\bibitem{Tan:08}
Weehong Tan, Ufuk Topcu, Pete Seiler, Gary Balas, and Andrew Packard.
\newblock Simulation-aided reachability and local gain analysis for nonlinear
  dynamical systems.
\newblock In {\em Proc. of the IEEE Conference on Decision and Control}, pages
  4097--4102. 2008.

\bibitem{Tobenkin:11}
Mark~M Tobenkin, Ian~R Manchester, and Russ Tedrake.
\newblock Invariant funnels around trajectories using sum-of-squares
  programming.
\newblock In {\em IFAC Proceedings}, volume~44, pages 9218--9223. 2011.

\bibitem{Topcu:09}
Ufuk Topcu and Andrew Packard.
\newblock Local robust performance analysis for nonlinear dynamical systems.
\newblock In {\em Proc. of American Control Conference}, pages 784--789.
  Shanghai, China, 2009.

\bibitem{Veenman:16}
Joost Veenman, Carsten Scherer, and Hakan Koroglu.
\newblock Robust stability and performance analysis based on integral quadratic
  constraints.
\newblock {\em European Journal of Control}, 31:1--32, 2016.

\bibitem{Bai:18}
Bai Xue, Martin Fr\"{a}nzle, and Naijun Zhan.
\newblock Under-approximating reach sets for polynomial continuous systems.
\newblock In {\em Proc. of the 21st International Conference on Hybrid Systems:
  Computation and Control}. Porto, Portugal, 2018.

\bibitem{xue2019inner}
Bai Xue, Martin Fr{\"a}nzle, and Naijun Zhan.
\newblock Inner-approximating reachable sets for polynomial systems with
  time-varying uncertainties.
\newblock {\em IEEE Transactions on Automatic Control}, 2019.

\bibitem{xue2019robust}
Bai Xue, Qiuye Wang, Naijun Zhan, and Martin Fr{\"a}nzle.
\newblock Robust invariant sets generation for state-constrained perturbed
  polynomial systems.
\newblock In {\em Proceedings of the 22nd ACM International Conference on
  Hybrid Systems: Computation and Control}, pages 128--137. ACM, 2019.

\bibitem{Zames:68}
G.~Zames and P.~L. Falb.
\newblock Stability conditions for systems with monotone and slope-restricted
  nonlinearities.
\newblock {\em SIAM Journal on Control}, 6(1):89--108, 1968.

\bibitem{Zhou:1996}
Kemin Zhou, John~C. Doyle, and Keith Glover.
\newblock {\em Robust and Optimal Control}.
\newblock Prentice-Hall, Inc., Upper Saddle River, NJ, USA, 1996.

\end{thebibliography}
\bibliographystyle{IEEEtran}  

\end{document}